\documentclass[9pt]{llncs}

\usepackage{tikz}
\usepackage{pgfplots}
\usetikzlibrary{decorations.markings}
\usetikzlibrary{matrix,backgrounds,folding}
\usetikzlibrary{shapes.geometric}
\pgfdeclarelayer{edgelayer}
\pgfdeclarelayer{nodelayer}
\pgfsetlayers{background,edgelayer,nodelayer,main}
\tikzstyle{every picture}=[baseline=-0.25em]
\tikzstyle{none}=[inner sep=0mm]
\tikzstyle{zxnode}=[shape=circle, minimum width=.25cm, inner sep=0.5pt, font=\footnotesize, draw=black]
\tikzstyle{gn}=[zxnode ,fill=green]
\tikzstyle{rn}=[zxnode ,fill=red]
\tikzstyle{H box}=[rectangle,fill=yellow,draw=black,xscale=1,yscale=1,font=\small,inner sep=0.75pt,minimum width=0.15cm,minimum height=0.15cm]
\tikzstyle{ug}=[regular polygon, regular polygon sides=3, fill=red,draw=black,inner sep = 0pt,minimum width=1em]
\tikzstyle{black dot}=[inner sep=0.7mm,minimum width=0pt,minimum height=0pt,fill=black,draw=black,shape=circle]
\tikzstyle{dot}=[black dot]
\tikzstyle{white dot}=[dot,fill=white]
\tikzstyle{zwcross}=[diamond, draw, fill=gray, minimum width=0em, inner sep=1.5pt]

\tikzstyle{st}=[star,star points = 5, fill=white,draw=black,inner sep = 1.2pt,line width=1.2pt]
\tikzstyle{uglabel}=[rounded corners=0.2em,fill=red!20,inner sep=0.1em,font=\scriptsize, anchor=west, xshift=0.1em, yshift=-0.2em,opacity=1]

\pgfdeclarelayer{edgelayer}
\pgfdeclarelayer{nodelayer}
\pgfsetlayers{background,edgelayer,nodelayer,main}
\tikzstyle{none}=[inner sep=0mm]
\tikzstyle{every loop}=[]

\def\fig{}

\usepackage{calc}
\usepackage{amssymb,nicefrac}
\usepackage{amsfonts}
\usepackage{mathtools}
\usepackage{multicol}
\usepackage{xspace}
\usepackage{dblfloatfix}
\usepackage{transparent}
\usepackage{stmaryrd}
\usepackage[colorlinks=true, citecolor=blue]{hyperref}

\newcommand{\eq}[2][~~]{
#1
\underset{\substack{#2}}{=}
#1
}

\newcommand{\equi}[2][\quad]{
#1
\underset{\substack{#2}}{\iff}
#1
}

\newcommand{\titlerule}[1]{\noindent\begin{minipage}{\columnwidth}\begin{center}
\rule{(\columnwidth-\widthof{#1})/2}{0.5pt}#1\rule{(\columnwidth-\widthof{#1})/2}{0.5pt}
\end{center}
\vspace{7em}
\end{minipage}
\vspace{-7.4em}}

\newcommand{\interp}[1] {\left\llbracket #1 \right\rrbracket}

\renewcommand{\arccos}[1]{\operatorname{arccos}\left(#1\right)}

\newcommand{\fit}[1] {\resizebox{\columnwidth}{!}{#1}}

\renewcommand{\implies}{\quad\Rightarrow\quad}
\renewcommand{\mod}{\bmod}
\newcommand{\frag}[1]{$\frac{\pi}{#1}$-fragment}
\renewcommand{\phi}{\varphi}
\newcommand{\annoted}[3]{\overbrace{#3}^{#2}\left.\vphantom{#3}\right\rbrace{\scriptstyle #1}}

\def \zx {\textnormal{ZX}\xspace}
\def \zxcalc {\text{ZX-Calculus}\xspace}

\def \zw {\textnormal{ZW}\xspace}

\def \rulecolor {red}
\newcommand{\callrule}[2]{\hyperlink{r:#1}{\textnormal{\color{\rulecolor}{(#2)}}}\xspace}

\newcommand{\so}[1]{\callrule{rules#1}{S1}}
\newcommand{\st}[1]{\callrule{rules#1}{S2}}
\newcommand{\sth}[1]{\callrule{rules#1}{S3}}
\newcommand{\e}[1]{\callrule{rules#1}{E}}
\newcommand{\bo}[1]{\callrule{rules#1}{B1}}
\newcommand{\bt}[1]{\callrule{rules#1}{B2}}
\newcommand{\kt}[1]{\callrule{rules#1}{K}}

\newcommand{\h}[1]{\callrule{rules#1}{H}}
\newcommand{\supp}[1]{\callrule{rules#1}{SUP}}
\newcommand{\com}[1]{\callrule{rules#1}{C}}

\newcommand{\nfi}[1]{\callrule{rules#1}{BW}}

\newcommand{\add}[1]{\callrule{rules#1}{A}}

\newcommand{\rz}[1]{~\begin{tikzpicture}
	\begin{pgfonlayer}{nodelayer}
		\node [style=gn] (0) at (0, 0.25) {#1};
		\node [style=none] (1) at (0, -0.25) {};
	\end{pgfonlayer}
	\begin{pgfonlayer}{edgelayer}
		\draw [style=none] (0) to (1.center);
	\end{pgfonlayer}
\end{tikzpicture}~}

\newcommand{\rx}[1]{~\begin{tikzpicture}
	\begin{pgfonlayer}{nodelayer}
		\node [style=rn] (0) at (0, 0.25) {#1};
		\node [style=none] (1) at (0, -0.25) {};
	\end{pgfonlayer}
	\begin{pgfonlayer}{edgelayer}
		\draw [style=none] (0) to (1.center);
	\end{pgfonlayer}
\end{tikzpicture}~}

\renewcommand*{\arraystretch}{0.8}
\allowdisplaybreaks

\title{Diagrammatic Reasoning beyond Clifford+T~Quantum~Mechanics} 
\author{
Emmanuel Jeandel
\and Simon Perdrix
\and Renaud Vilmart
\institute{Universit\'e de Lorraine, CNRS, Inria, LORIA, F 54000 Nancy, France}
\\\email{emmanuel.jeandel@loria.fr}$\quad$
\email{simon.perdrix@loria.fr}$\quad$
\email{renaud.vilmart@loria.fr}
}

\begin{document}   

\maketitle
\begin{abstract}
The ZX-Calculus is  a graphical language for diagrammatic reasoning in quantum mechanics and quantum information theory. An axiomatisation has recently been proven to be complete for an approximatively universal fragment of quantum mechanics, the so-called Clifford+T fragment.
We focus here on the expressive power of this axiomatisation beyond  Clifford+T Quantum mechanics. We consider the full pure qubit quantum mechanics, and mainly prove two results: (i) First, the axiomatisation for Clifford+T quantum mechanics is also complete for all equations involving some kind of \emph{linear} diagrams.
The linearity of the diagrams reflects the phase group structure, an essential feature of the ZX-calculus. In particular all the axioms of the ZX-calculus are involving linear diagrams. 
 (ii) We also show that the axiomatisation for Clifford+T is not complete in general but can be completed by adding a single (non linear) axiom, providing a simpler axiomatisation of the ZX-calculus for pure quantum mechanics than the one recently introduced by Ng\&Wang.
\end{abstract}

\section{Introduction}
The ZX-calculus, introduced by Coecke and Duncan \cite{interacting} is a graphical language for pure state qubit quantum mechanics. The ZX-calculus has multiple applications in quantum information theory \cite{picturing-qp}, including the foundations \cite{toy-model-graph,duncan2016hopf}, measurement-based quantum computation \cite{mbqc,horsman2011quantum,duncan2013mbqc} or quantum error correcting codes \cite{verifying-color-code,duncan2014verifying,chancellor2016coherent,de2017zx}, and can be used through the interactive theorem prover Quantomatic \cite{quanto,kissinger2015quantomatic}.

The ZX-calculus is universal: any quantum evolution can be represented by a ZX-diagram. ZX-diagrams are parametrised by angles, and various fragments of the language have been considered, based on some restrictions on the angles: the $\frac \pi p$-fragment consists in considering  only the diagrams made with angles multiple of $\frac \pi p$. The $\frac \pi 2$-fragment  (resp. $\pi$-) corresponds to stabilizer quantum mechanics (resp. real stabilizer quantum mechanics) and are not universal for quantum mechanics, even approximately. The $\frac \pi 4$-fragment corresponds to the so called Clifford+T quantum mechanics and is approximately universal: any quantum evolution can be approximated in this fragment with arbitrary accuracy. 

The ZX-calculus also comes with a powerful axiomatisation which can be used to transform a diagram into another diagram representing the same quantum evolution. The axioms of the ZX-calculus are given in Figure \ref{fig:ZX_rules}. Some of the axioms are parametrised by variables, meaning that the axioms are true for all possible values of these variables. Notice that all the variables are used in a linear fashion, i.e. all the angles are some linear combinations of variables and constants, like in \so{} or \supp{} for instance.  The use of such linear diagrams in the axiomatisation captures the phase group structure, one of the two fundamental quantum features (with the complementary observables) of the ZX-calculus \cite{interacting}. 

Completeness of the axiomatisation is an essential feature: the axiomatisation is complete if for any pair of diagrams representing the same quantum evolution, one can use the axioms of the language to transform one diagram into the other. The ZX-calculus has been proved to be complete for the $\pi$- and $\frac \pi 2$-fragments of the ZX-calculus \cite{pivoting,pi_2-complete}. Recently the axiomatisation given in Figure \ref{fig:ZX_rules} has been proved to be completed for the \frag 4, providing the first complete axiomatisation for an approximately universal fragment of the ZX-calculus \cite{JPV}. 
This last result relies on the completeness of another graphical language which represents integer matrices, called ZW-Calculus \cite{zw}. The ZW-Calculus has since been extended to represent all matrices over $\mathbb{C}$ \cite{Amar}. This achievement gave hope for a universal completion of the ZX-Calculus, and soon enough, a first result appeared \cite{NgWang}. To make the ZX-calculus complete for the full quantum mechanics, two new generators and a large amount of axioms (32 axioms versus 12 for the axiomatisation for Clifford+T quantum mechanics) have been introduced, some of them being non linear.  

One can wonder whether this result can be improved. We address this question in two steps: (i) First, we prove that the complete axiomatisation for Clifford+T quantum mechanics can also be used to prove a significant amount of equations beyond this fragment: all true equations involving diagrams which are linear with constants multiple of $\frac \pi 4$ can be derived. We point out with several examples that this result can be used to derive some new non-trivial equations. (ii) Then we show that this axiomatisation is not complete in general, and we propose an axiomatisation for the full pure qubit quantum mechanics which consists in adding a single (non-linear) axiom.

The paper is structured as follows. The ZX-calculus is presented in section \ref{sec:zx-pi_4}. Section \ref{sec:thm-param-diag} is dedicated to the proof that any true equation involving diagrams linear in some variables with constants multiple of $\frac \pi 4$ can be derived in the ZX-calculus. In sections \ref{sec:finite-case-based-reasoning} and \ref{sec:diagram-substitution} we show how this result can be used to prove that some non trivial equations can be derived in the ZX-calculus, in a non-necessarily constructive way. Section \ref{sec:completion} is dedicated to the completion of the ZX-calculus for the full pure qubit quantum mechanics: first, we prove that the ZX-calculus is not complete for pure qubit quantum mechanics; then, using an interpretation from the ZX-calculus to the ZW-Calculus we show that a single additional axiom suffices to make the language complete.

\section{ZX-Calculus}

\label{sec:zx-pi_4}
\subsection{Syntax and Semantics}

A ZX-diagram $D:k\to l$ with $k$ inputs and $l$ outputs is generated by:\\
\begin{center}
\bgroup
\def\arraystretch{2.5}
{\begin{tabular}{|cc|cc|}
\hline
$R_Z^{(n,m)}(\alpha):n\to m$ & 
\InputIfFileExists{gn-alpha.tikz}{}{\input{./figures/gn-alpha.tikz}}
 & $R_X^{(n,m)}(\alpha):n\to m$ & 
\InputIfFileExists{rn-alpha.tikz}{}{\input{./figures/rn-alpha.tikz}}
\\[4ex]\hline
$H:1\to 1$ & 
\begin{tikzpicture}
	\begin{pgfonlayer}{nodelayer}
		\node [style={H box}] (0) at (0, 0) {};
		\node [style=none] (1) at (0, 0.5) {};
		\node [style=none] (2) at (0, -0.5) {};
	\end{pgfonlayer}
	\begin{pgfonlayer}{edgelayer}
		\draw (2.center) to (1.center);
	\end{pgfonlayer}
\end{tikzpicture}
}
 & $e:0\to 0$ & 
\InputIfFileExists{empty-diagram.tikz}{}{\input{./figures/empty-diagram.tikz}}
\\\hline
$\mathbb{I}:1\to 1$ & 
\begin{tikzpicture}
	\begin{pgfonlayer}{nodelayer}
		\node [style=none] (0) at (0, 0.2499999) {};
		\node [style=none] (1) at (0, -0.2499999) {};
	\end{pgfonlayer}
	\begin{pgfonlayer}{edgelayer}
		\draw (0.center) to (1.center);
	\end{pgfonlayer}
\end{tikzpicture}}
 & $\sigma:2\to 2$ & 
\InputIfFileExists{crossing.tikz}{}{\input{./figures/crossing.tikz}}
\\\hline
$\epsilon:2\to 0$ & 
\begin{tikzpicture}
	\begin{pgfonlayer}{nodelayer}
		\node [style=none] (0) at (-0.2500001, 0.2500001) {};
		\node [style=none] (1) at (0.2500001, 0.2500001) {};
	\end{pgfonlayer}
	\begin{pgfonlayer}{edgelayer}
		\draw [bend right=90, looseness=1.75] (0.center) to (1.center);
	\end{pgfonlayer}
\end{tikzpicture}}
 & $\eta:0\to 2$ & 
\begin{tikzpicture}
	\begin{pgfonlayer}{nodelayer}
		\node [style=none] (0) at (-0.2500001, -0) {};
		\node [style=none] (1) at (0.2500001, -0) {};
	\end{pgfonlayer}
	\begin{pgfonlayer}{edgelayer}
		\draw [bend left=90, looseness=1.75] (0.center) to (1.center);
	\end{pgfonlayer}
\end{tikzpicture}}
\\\hline
\end{tabular}}
\egroup\\
where $n,m\in \mathbb{N}$ and $\alpha \in \mathbb{R}$. The generator $e$ is the empty diagram.
\end{center}

\vspace{0.2cm}
and the two compositions:
\begin{itemize}
\item Spacial Composition: for any $D_1:a\to b$ and $D_2:c\to d$, $D_1\otimes D_2:a+c\to b+d$ consists in placing $D_1$ and $D_2$ side by side, $D_2$ on the right of $D_1$.
\item Sequential Composition: for any $D_1:a\to b$ and $D_2:b\to c$, $D_2\circ D_1:a\to c$ consists in placing $D_1$ on the top of $D_2$, connecting the outputs of $D_1$ to the inputs of $D_2$.
\end{itemize}

The standard interpretation of the ZX-diagrams associates to any diagram $D:n\to m$ a linear map $\interp{D}:\mathbb{C}^{2^n}\to\mathbb{C}^{2^m}$ inductively defined as follows:\\
\titlerule{$\interp{.}$}
\[ \interp{D_1\otimes D_2}:=\interp{D_1}\otimes\interp{D_2} \qquad 
\interp{D_2\circ D_1}:=\interp{D_2}\circ\interp{D_1}\]
\[\interp{
\InputIfFileExists{empty-diagram.tikz}{}{\input{./figures/empty-diagram.tikz}}
~}:=\begin{pmatrix}
1
\end{pmatrix} \qquad
\interp{~
\begin{tikzpicture}
	\begin{pgfonlayer}{nodelayer}
		\node [style=none] (0) at (0, 0.2499999) {};
		\node [style=none] (1) at (0, -0.2499999) {};
	\end{pgfonlayer}
	\begin{pgfonlayer}{edgelayer}
		\draw (0.center) to (1.center);
	\end{pgfonlayer}
\end{tikzpicture}}
~~}:= \begin{pmatrix}
1 & 0 \\ 0 & 1\end{pmatrix}\qquad
\interp{~
\begin{tikzpicture}
	\begin{pgfonlayer}{nodelayer}
		\node [style={H box}] (0) at (0, 0) {};
		\node [style=none] (1) at (0, 0.5) {};
		\node [style=none] (2) at (0, -0.5) {};
	\end{pgfonlayer}
	\begin{pgfonlayer}{edgelayer}
		\draw (2.center) to (1.center);
	\end{pgfonlayer}
\end{tikzpicture}
}
~}:= \frac{1}{\sqrt{2}}\begin{pmatrix}1 & 1\\1 & -1\end{pmatrix}\]
$$\interp{
\InputIfFileExists{crossing.tikz}{}{\input{./figures/crossing.tikz}}
}:= \begin{pmatrix}
1&0&0&0\\
0&0&1&0\\
0&1&0&0\\
0&0&0&1
\end{pmatrix} \qquad
\interp{\raisebox{-0.25em}{$
\begin{tikzpicture}
	\begin{pgfonlayer}{nodelayer}
		\node [style=none] (0) at (-0.2500001, 0.2500001) {};
		\node [style=none] (1) at (0.2500001, 0.2500001) {};
	\end{pgfonlayer}
	\begin{pgfonlayer}{edgelayer}
		\draw [bend right=90, looseness=1.75] (0.center) to (1.center);
	\end{pgfonlayer}
\end{tikzpicture}}
$}}:= \begin{pmatrix}
1&0&0&1
\end{pmatrix} \qquad
\interp{\raisebox{-0.35em}{$
\begin{tikzpicture}
	\begin{pgfonlayer}{nodelayer}
		\node [style=none] (0) at (-0.2500001, -0) {};
		\node [style=none] (1) at (0.2500001, -0) {};
	\end{pgfonlayer}
	\begin{pgfonlayer}{edgelayer}
		\draw [bend left=90, looseness=1.75] (0.center) to (1.center);
	\end{pgfonlayer}
\end{tikzpicture}}
$}}:= \begin{pmatrix}
1\\0\\0\\1
\end{pmatrix}$$
For any $\alpha\in\mathbb{R}$, $\interp{\begin{tikzpicture}
	\begin{pgfonlayer}{nodelayer}
		\node [style=gn] (0) at (0, -0) {$\alpha$};
	\end{pgfonlayer}
\end{tikzpicture}}:=\begin{pmatrix}1+e^{i\alpha}\end{pmatrix}$, and  for any $n,m\geq 0$ such that $n+m>0$:\vspace{-1em}
$$
\interp{
\InputIfFileExists{gn-alpha.tikz}{}{\input{./figures/gn-alpha.tikz}}
}:=
\annoted{2^m}{2^n}{\begin{pmatrix}
  1 & 0 & \cdots & 0 & 0 \\
  0 & 0 & \cdots & 0 & 0 \\
  \vdots & \vdots & \ddots & \vdots & \vdots \\
  0 & 0 & \cdots & 0 & 0 \\
  0 & 0 & \cdots & 0 & e^{i\alpha}
 \end{pmatrix}}
$$
\begin{minipage}{\columnwidth}
$$\interp{
\InputIfFileExists{rn-alpha.tikz}{}{\input{./figures/rn-alpha.tikz}}
}:=\interp{~
\begin{tikzpicture}
	\begin{pgfonlayer}{nodelayer}
		\node [style={H box}] (0) at (0, 0) {};
		\node [style=none] (1) at (0, 0.5) {};
		\node [style=none] (2) at (0, -0.5) {};
	\end{pgfonlayer}
	\begin{pgfonlayer}{edgelayer}
		\draw (2.center) to (1.center);
	\end{pgfonlayer}
\end{tikzpicture}
}
~}^{\otimes m}\circ \interp{
\InputIfFileExists{gn-alpha.tikz}{}{\input{./figures/gn-alpha.tikz}}
}\circ \interp{~
\begin{tikzpicture}
	\begin{pgfonlayer}{nodelayer}
		\node [style={H box}] (0) at (0, 0) {};
		\node [style=none] (1) at (0, 0.5) {};
		\node [style=none] (2) at (0, -0.5) {};
	\end{pgfonlayer}
	\begin{pgfonlayer}{edgelayer}
		\draw (2.center) to (1.center);
	\end{pgfonlayer}
\end{tikzpicture}
}
~}^{\otimes n}$$ \\
$\left(\text{where }M^{\otimes 0}=\begin{pmatrix}1\end{pmatrix}\text{ and }M^{\otimes k}=M\otimes M^{\otimes k-1}\text{ for any }k\in \mathbb{N}^*\right)$.\\
\rule{\columnwidth}{0.5pt}
\end{minipage}\\

To simplify, the red and green nodes will be represented empty when holding a 0 angle:
\[ 
\InputIfFileExists{gn-empty-is-gn-zero.tikz}{}{\input{./figures/gn-empty-is-gn-zero.tikz}}
 \qquad\text{and}\qquad 
\InputIfFileExists{rn-empty-is-rn-zero.tikz}{}{\input{./figures/rn-empty-is-rn-zero.tikz}}
 \]

\subsection{Complete axiomatisation for Clifford+T}

The complete axiomatisation of the ZX-calculus for Clifford+T introduced in \cite{JPV} is given in Figure \ref{fig:ZX_rules}.

\begin{figure*}[!htb]
 \centering
 \hypertarget{r:rules}{}
 \begin{tabular}{|c@{$\qquad$}c|}
   \hline
   & \\
   
\InputIfFileExists{spider-1.tikz}{}{\input{./figures/spider-1.tikz}}
& 
\InputIfFileExists{s2-simple.tikz}{}{\input{./figures/s2-simple.tikz}}
\\
   & \\
   
\InputIfFileExists{induced_compact_structure-2wire.tikz}{}{\input{./figures/induced_compact_structure-2wire.tikz}}
&
\InputIfFileExists{bicolor_pi_4_eq_empty.tikz}{}{\input{./figures/bicolor_pi_4_eq_empty.tikz}}
\\
   & \\
   
\InputIfFileExists{b1s.tikz}{}{\input{./figures/b1s.tikz}}
& 
\InputIfFileExists{b2s.tikz}{}{\input{./figures/b2s.tikz}}
\\
   & \\
   
\InputIfFileExists{euler-decomp-scalar-free.tikz}{}{\input{./figures/euler-decomp-scalar-free.tikz}}
&  
\InputIfFileExists{h2.tikz}{}{\input{./figures/h2.tikz}}
\\
   & \\
   
\InputIfFileExists{k2s.tikz}{}{\input{./figures/k2s.tikz}}
& 
\InputIfFileExists{former-supp.tikz}{}{\input{./figures/former-supp.tikz}}
 \\
   & \\
   
\InputIfFileExists{commutation-of-controls-general-simplified.tikz}{}{\input{./figures/commutation-of-controls-general-simplified.tikz}}
&
\InputIfFileExists{BW-simplified.tikz}{}{\input{./figures/BW-simplified.tikz}}
\\
   & \\
   \hline
  \end{tabular}
 \caption[]{Set of rules for the Clifford+T \zxcalc with scalars. All of these rules also hold when flipped upside-down, or with the colours red and green swapped. The right-hand side of (E) is an empty diagram. (...) denote zero or more wires, while (\protect\rotatebox{45}{\raisebox{-0.4em}{$\cdots$}}) denote one or more wires.
 }
 \label{fig:ZX_rules}
\end{figure*}

These rules come together with a set of implicit axioms aggregated under the paradigm ``Only Topology Matters'', which states that the way the wires are bent or cross each other does not matter. What only matters is whether two dots are connected or not. Such rules are:
\[\fit{
\InputIfFileExists{bent-wire.tikz}{}{\input{./figures/bent-wire.tikz}}
}\]
\[{
\InputIfFileExists{bent-wire-2.tikz}{}{\input{./figures/bent-wire-2.tikz}}
}\]

The equality between diagrams is preserved when axioms are applied locally, which means that for any three diagrams of the ZX-Calculus, $D_1, D_2$, and $D$, if $\zx\vdash D_1=D_2$, then:\\
\renewcommand*{\arraystretch}{1.2}
\begin{tabular}{@{$\qquad$}l@{$\qquad$}l}
$\bullet\zx\vdash D_1\circ D = D_2\circ D$ & $\bullet \zx\vdash D\circ D_1 = D\circ D_2$\\
$\bullet\zx\vdash D_1\otimes D = D_2\otimes D$ & $\bullet \zx\vdash D\otimes D_1 = D\otimes D_2$
\end{tabular}\\
\renewcommand*{\arraystretch}{0.8}
where $\zx\vdash D_1 = D_2$ means that $D_1$ can be transformed into $D_2$ using the axioms of the ZX-Calculus. 

Notice that some rules are specific to the $\frac \pi 4$ angle, like \e{} or \nfi{}, whereas some others, \so{}, \h{}, \kt{}, \supp{} and \com{} are parametrised by angles that can take whatever value in $\mathbb{R}$. In the following, ZX will denote either the set of general diagrams (with angles in $\mathbb{R}$) or the set of general rules in Figure \ref{fig:ZX_rules}.

\subsection{Variables and Constants}

It is customary to view some angles in the ZX-diagrams as variables, in order to prove families of equalities. For instance, the rule \so{} displays two variables $\alpha$ and $\beta$, and potentially gives an infinite number of equalities. Notice that in the axioms of the ZX-calculus, the variables are used in a linear way, reflecting the phase group structure. 

\begin{definition} 
A ZX-diagram is linear in $\alpha_1, \ldots, \alpha_k$ with constants in $C\subseteq \mathbb R$, if it is generated by $R_Z^{(n,m)}(E)$, $R_X^{(n,m)}(E)$, $H$, $e$, $\mathbb I$, $\sigma$, $\epsilon$, $\eta$, and the spacial and sequential compositions, where $n,m\in \mathbb  N$, and $E$ is of the form $\sum_{i} n_i \alpha_i+c$, with $n_i\in \mathbb Z$ and $c\in C$. 
\end{definition}

Notice that all the diagrams in Figure \ref{fig:ZX_rules} are linear in $\alpha, \beta, \gamma$ with constants in $\frac \pi 4 \mathbb Z$. A diagram linear in $\alpha_1, \ldots,  \alpha_k$ is denoted $D(\alpha_1, \ldots, \alpha_k)$, or more compactly $D(\vec \alpha)$ with $\vec \alpha = \alpha_1, \ldots, \alpha_k$. Obviously, if $D(\alpha)$ is a diagram linear in $\alpha$, $D(\pi/2)$ denotes the ZX-diagram where all occurrences of $\alpha$ are replaced by $\pi/2$.

\section{Proving Equalities beyond Clifford+T}
\label{sec:thm-param-diag}

While the set of rules of Figure \ref{fig:ZX_rules} is complete for the Clifford+T fragment of the ZX-calculus, it can also prove a lot of equalities for the general ZX-calculus, when the rules \so{}, \h{}, \kt{}, \com{} are supposed to hold for all angles rather than angles in the \frag4.

In fact, it can prove all equalities that are valid for linear diagrams with constants multiple of $\frac \pi 4$, in the following sense:

\begin{theorem}
\label{thm:provability}  
For any ZX-diagrams $D_1(\vec \alpha)$ and $D_2(\vec \alpha)$ linear in $\vec \alpha=\alpha_1, \ldots, \alpha_k$ with constants in $\frac \pi 4 \mathbb Z$, 
$$  \forall \vec \alpha \in \mathbb R^k, \interp{D_1(\vec \alpha)}= \interp{D_2(\vec \alpha)}       \,\,\,\Leftrightarrow\,\,\,  \forall \vec \alpha \in \mathbb R^k, \zx\vdash D_1(\vec \alpha) = D_2(\vec \alpha)       $$
\end{theorem}

%
%
%
%
%
%

The proof essentially relies on the completeness of the $\pi/4$-fragment of the ZX-calculus: the variables are first turned into inputs of the diagrams (Prop. \ref{prop:var2inp} and \ref{prop:vars2inp}) and then replaced by some constant diagram in the \frag4 (Lem. \ref{lem:equivalence-X} and \ref{lem:equivalence-Pk}). To simplify the proofs, we will first consider the case where a single variable -- with potentially several occurrences --  is involved in the equation, the general case being similar and addressed in section \ref{sec:multiple-variables}. 

\subsection{From variables to inputs}

We show in this section that, given an equation involving diagrams linear in some variable $\alpha$, the variables can be \emph{extracted}, splitting  the diagrams into two parts: a collection of points (points $\alpha$) and a constant diagram independent of the variables. 

First we define the multiplicity of a variable in an equation:

\begin{definition}
For any $D_1(\alpha), D_2(\alpha): n\to m$ two ZX-diagrams linear in $\alpha$, the multiplicity of $\alpha$ in the equation $D_1(\alpha) = D_2(\alpha)$ is defined as:
$$\mu_\alpha = \max_{i\in \{1,2\}}\left(\mu^+_\alpha(D_i(\alpha))\right)  + \max_{i\in \{1,2\}}\left(\mu^-_\alpha(D_i(\alpha))\right)$$
 where 
$\mu^+_\alpha(D)$ (resp. $\mu^-_\alpha(D)$) is the number of occurrences of $\alpha$ (resp. -$\alpha$) in $D$,  inductively defined as   \\
$\mu^+_\alpha(R_Z^{(n,m)}(\ell\alpha+c))=\mu^+_\alpha(R_X^{(n,m)}(\ell\alpha+c))=\begin{cases} \ell &\text{if $\ell>0$}\\0&\text{otherwise}\end{cases}$\\
$\mu^-_\alpha(R_Z^{(n,m)}(\ell\alpha+c))=\mu^-_\alpha(R_X^{(n,m)}(\ell\alpha+c))=\begin{cases} -\ell&\text{if $\ell<0$}\\0&\text{otherwise}\end{cases}$\\
$\forall \diamond \in \{+,-\}$, 
$\mu^\diamond_\alpha(D\otimes D') = \mu^\diamond_\alpha(D\circ D') = \mu^\diamond_\alpha(D)+\mu^\diamond_\alpha(D')$\\
$\mu^\diamond_\alpha(H)=\mu^\diamond_\alpha(e)=\mu^\diamond_\alpha(\mathbb I)=\mu^\diamond_\alpha(\sigma)=\mu^\diamond_\alpha(\epsilon)=\mu^\diamond_\alpha(\eta)=0$

\end{definition}



\begin{proposition}\label{prop:var2inp}
For any $D_1(\alpha), D_2(\alpha): n\to m$ two ZX-diagrams linear in $\alpha$ with constants in $\frac \pi 4$ $\mathbb Z$, there exist  $D_1', D'_2:r\to n+m$ two ZX-diagrams with angles multiple of $\frac \pi 4$ such that, for any $\alpha \in \mathbb R$, the equivalence 
\begin{equation}\label{eq:varasinputs}
\zx\vdash D_1(\alpha)= D_2(\alpha) \equi[\ ]{}\zx\vdash D_1'\circ \theta_r(\alpha) =  D_2'\circ \theta_r(\alpha)
\end{equation}
is provable using the axioms of the ZX-calculus, where $r$ is the multiplicity of $\alpha$ in $D_1(\alpha) = D_2(\alpha)$, and $\theta_r(\alpha)= \left(R_Z^{(0,1)}(\alpha)\right)^{\otimes r}$.\\
Pictorially: 
\def\fig{theta_r-alpha-on-diagrams}
\[\zx\vdash\begin{tikzpicture}
	\begin{pgfonlayer}{nodelayer}
		\node [style=dot] (39)  at (-0.125, 0.5) {};
		\node [style=dot] (40)  at (-0.125, 0.25) {};
		\node [style=white dot] (41)  at (-0.125, -0.0) {};
		\node [style=none, anchor=west] (42)  at (0.125, -0.0) {\footnotesize $\rho e^{i\theta}$};
		\node [style=none] (43)  at (-0.125, -0.5) {};
	\end{pgfonlayer}
	\begin{pgfonlayer}{edgelayer}
		\draw (39) to (43.center);
		\draw [style=none, in=135, out=45, loop] (39) to ();
	\end{pgfonlayer}
\end{tikzpicture}\eq{}\begin{tikzpicture}
	\begin{pgfonlayer}{nodelayer}
		\node [style=gn] (0)  at (-1.25, -0.5) {-$\gamma$};
		\node [style=gn] (1)  at (-1.25, 0.5) {$\gamma$};
		\node [style=rn] (2)  at (-1.0, -0.0) {$\pi$};
		\node [style=gn] (3)  at (-0.5, -0.0) {$\theta$};
		\node [style=rn] (4)  at (-0.5, 0.5) {};
		\node [style=none] (5)  at (-0.5, -0.75) {};
		\node [style=rn] (6)  at (0.0, -0.0) {};
		\node [style=gn] (7)  at (0.25, -0.5) {-$\beta$};
		\node [style=gn] (8)  at (0.25, 0.5) {$\beta$};
		\node [style=none] (9)  at (1.25, 0.5) {$)$};
		\node [style=none] (10)  at (0.75, 0.5) {$($};
		\node [style=gn] (11)  at (1.0, 0.5) {};
		\node [style=none, anchor=west] (12)  at (1.25, 0.75) {\scriptsize $\otimes n$};
		\node [style=gn] (13)  at (1.0, -0.75) {};
		\node [style=none, xshift=4pt, anchor=east] (14)  at (0.75, -0.5) {$\left(\vphantom{\rule{1pt}{1.5em}}\right.$};
		\node [style=none, anchor=west] (15)  at (1.25, -0.0) {\scriptsize $\otimes 3$};
		\node [style=rn] (16)  at (1.0, -0.25) {};
		\node [style=none, anchor=west, xshift=-4pt] (17)  at (1.25, -0.5) {$\left.\vphantom{\rule{1pt}{1.5em}}\right)$};
	\end{pgfonlayer}
	\begin{pgfonlayer}{edgelayer}
		\draw (1) to (2);
		\draw (2) to (0);
		\draw (2) to (6);
		\draw (4) to (5.center);
		\draw (6) to (7);
		\draw (8) to (6);
		\draw [bend right=45, looseness=1.00] (13) to (16);
		\draw [bend right=45, looseness=1.00] (16) to (13);
		\draw (16) to (13);
	\end{pgfonlayer}
\end{tikzpicture}\equi[\ ]{}\zx\vdash\begin{tikzpicture}
	\begin{pgfonlayer}{nodelayer}
		\node [style=none] (19)  at (-0.125, -0.75) {};
		\node [style=none, anchor=west] (20)  at (0.875, 0.75) {\scriptsize $\otimes n$};
		\node [style=rn] (21)  at (0.625, -0.25) {};
		\node [style=none] (22)  at (0.375, 0.5) {$($};
		\node [style=gn] (23)  at (0.625, 0.5) {};
		\node [style=gn] (24)  at (0.625, -0.75) {};
		\node [style=none, xshift=-4pt, anchor=west] (25)  at (0.875, -0.5) {$\left.\vphantom{\rule{1pt}{1.5em}}\right)$};
		\node [style=none] (26)  at (0.875, 0.5) {$)$};
		\node [style=none, anchor=east, xshift=4pt] (27)  at (0.375, -0.5) {$\left(\vphantom{\rule{1pt}{1.5em}}\right.$};
		\node [style=none, anchor=west] (28)  at (0.875, -0.0) {\scriptsize $\otimes 5$};
		\node [style=rn] (29)  at (-0.125, -0.25) {};
		\node [style=gn] (30)  at (-0.875, -0.5) {};
		\node [style=rn] (31)  at (-0.625, 0.25) {$\pi$};
		\node [style=gn] (32)  at (-0.875, 0.75) {$\gamma$};
		\node [style=gn] (33)  at (-0.125, 0.75) {-$\gamma$};
	\end{pgfonlayer}
	\begin{pgfonlayer}{edgelayer}
		\draw [bend right=45, looseness=1.00] (21) to (24);
		\draw (21) to (24);
		\draw [bend right=45, looseness=1.00] (24) to (21);
		\draw (29) to (19.center);
		\draw (31) to (33);
		\draw (32) to (31);
	\end{pgfonlayer}
\end{tikzpicture}\eq{}\begin{tikzpicture}
	\begin{pgfonlayer}{nodelayer}
		\node [style=gn] (34)  at (0.25, -0.125) {};
		\node [style=rn] (35)  at (-0.25, 0.125) {};
		\node [style=none] (36)  at (-0.25, -0.375) {};
		\node [style=rn] (38)  at (0.25, 0.375) {};
	\end{pgfonlayer}
	\begin{pgfonlayer}{edgelayer}
		\draw [bend right=45, looseness=1.00] (34) to (38);
		\draw (35) to (36.center);
		\draw [bend right=45, looseness=1.00] (38) to (34);
		\draw (38) to (34);
	\end{pgfonlayer}
\end{tikzpicture}\]
\end{proposition}

\begin{proof}The proof consists in transforming the equation $D_1(\alpha)=D_2(\alpha)$ into the equivalent equation $D_1'\circ \theta_r(\alpha) =  D_1'\circ \theta_r(\alpha)$ using axioms of the ZX-calculus. This transformation involves 6 steps:
\\-- \emph{Turn inputs into outputs.} First, each input can be bent to an output using $\eta$: 
\def\fig{thm1-equivalence}\[\eq{}\equi[]{}\eq{}\]
\\-- \emph{Make the red spiders green.} All red spiders $R^{(k,l)}_X(n\alpha+c)$ are transformed into green spiders using the axioms \so{} and \h{}:
\def\fig{h-on-red-spiders}
\[\eq{}\]
 \\-- \emph{Expending spiders.} All spiders $R_Z(n\alpha +c)$ are expended using \so{} so that all the occurrences of $\alpha$ are \rz{$\alpha$} or \rz{-$\alpha$}:
\def\fig{S1-on-n-alpha-plus-c-arxiv}
\[\eq{}\]
\\-- \emph{Changing the sign.} Using \kt{} all occurrences of \rz{-$\alpha$}
 are replaced as follows: $\rz{-$\alpha$} \mapsto
\InputIfFileExists{Rz-0-1-minus-alpha.tikz}{}{\input{./figures/Rz-0-1-minus-alpha.tikz}}
$. Notice that this rule is not applied recursively, which would not terminate. After this step all the original $-\alpha$ have been replaced by an $\alpha$ and as many scalars 
\InputIfFileExists{scalar-e-pow-minus-i-alpha.tikz}{}{\input{./figures/scalar-e-pow-minus-i-alpha.tikz}}
 have been created. So far, we have shown:
\def\fig{thm1-equivalence-one-var}
\begin{align*}
\eq{}\equi[\ \ ]{}
\input{./figures/\fig/\fig_06.tikz}\eq[]{}\input{./figures/\fig/\fig_07.tikz}
\end{align*}
\\-- \emph{(Re)moving scalars.} The scalar 
\InputIfFileExists{scalar-e-pow-i-alpha.tikz}{}{\input{./figures/scalar-e-pow-i-alpha.tikz}}
 has an inverse for $\otimes$, which is 
\InputIfFileExists{scalar-e-pow-minus-i-alpha.tikz}{}{\input{./figures/scalar-e-pow-minus-i-alpha.tikz}}
 (see Lemmas \ref{lem:multiplying-global-phases}, \ref{lem:bicolor-0-alpha} and \ref{lem:inverse}). This has for consequence:
\begin{itemize}
\item $\zx\vdash \scalebox{0.8}{
\InputIfFileExists{scalar-e-pow-minus-i-alpha.tikz}{}{\input{./figures/scalar-e-pow-minus-i-alpha.tikz}}
}D_1=D_2 \equi{}\zx\vdash D_1=\scalebox{0.8}{
\InputIfFileExists{scalar-e-pow-i-alpha.tikz}{}{\input{./figures/scalar-e-pow-i-alpha.tikz}}
}D_2$
\item $\zx\vdash \scalebox{0.8}{
\InputIfFileExists{scalar-e-pow-i-alpha.tikz}{}{\input{./figures/scalar-e-pow-i-alpha.tikz}}
}D_1=\scalebox{0.8}{
\InputIfFileExists{scalar-e-pow-i-alpha.tikz}{}{\input{./figures/scalar-e-pow-i-alpha.tikz}}
}D_2 \equi{} \zx\vdash D_1=D_2$
\end{itemize}
The scalars 
\InputIfFileExists{scalar-e-pow-minus-i-alpha.tikz}{}{\input{./figures/scalar-e-pow-minus-i-alpha.tikz}}
 are eliminated by adding $\max\limits_{i\in \{1,2\}}\left(\mu^-_\alpha(D_i)\right)$ times the scalar 
\InputIfFileExists{scalar-e-pow-i-alpha.tikz}{}{\input{./figures/scalar-e-pow-i-alpha.tikz}}
 on both sides, then simplifying when we have a scalar and its inverse.
\begin{align*}
\equi{}\input{./figures/\fig/\fig_08.tikz}\eq{}\input{./figures/\fig/\fig_09.tikz}
\end{align*}
\\-- \emph{Balancing the variables.} At this step the number of occurrences of $\alpha$ might be different on both sides of the equation. Indeed, one can check that the side of $D_i$ has $\mu^+_\alpha(D_i)+\max\limits_{j\in \{1,2\}}\left(\mu^-_\alpha(D_j)\right)$ occurrences of $\alpha$. One can then use the simple equation 
\InputIfFileExists{inverse-alpha.tikz}{}{\input{./figures/inverse-alpha.tikz}}
 (whose proof uses Lemmas \ref{lem:bicolor-0-alpha} and \ref{lem:inverse}) $\max\limits_{j\in \{1,2\}}\left(\mu^+_\alpha(D_j)\right)-\mu^+_\alpha(D_i)$ times on the side of $D_i$. We hence end up with $\mu_\alpha = \max\limits_{i\in \{1,2\}}\left(\mu^+_\alpha(D_i(\alpha))\right)  + \max\limits_{i\in \{1,2\}}\left(\mu^-_\alpha(D_i(\alpha))\right)$ occurrences of $\alpha$ on both sides. $D_i'$ is defined as:
\begin{align*}
\input{./figures/\fig/\fig_10.tikz}:=\input{./figures/\fig/\fig_11.tikz}
\end{align*}
\end{proof}

Proposition \ref{prop:var2inp}  implies in particular that if the equation $D_1'\circ \theta_r(\alpha) =  D_2'\circ \theta_r(\alpha)$ is provable  using the axioms of the ZX-calculus, then so is $D_1(\alpha)= D_2(\alpha)$.  
 Proposition \ref{prop:var2inp} also implies that if $\interp{D_1(\alpha)}= \interp{D_2(\alpha)}$, then $ \interp{D_1'\circ \theta_r(\alpha)} = \interp{ D_2'\circ \theta_r(\alpha)}$, thanks to the soundness of the ZX-calculus. 

%
%

 \subsection{Removing the variables}
 
 Given $D_1(\alpha)$ and $D_2(\alpha)$  linear in $\alpha$ with constants in $\frac \pi 4\mathbb Z$, if $\alpha$ has multiplicity $1$ in  ${D_1(\alpha)} = {D_2(\alpha)}$, then according to Prop. \ref{prop:var2inp}, the equation  can be transformed into the following equivalent equation involving a single occurrence of $\alpha$: 
\begin{equation}\label{eq:single}

\InputIfFileExists{thm1-single-occurrence.tikz}{}{\input{./figures/thm1-single-occurrence.tikz}}

\end{equation} 
 where $D_1'$ and $D_2'$ are in the $\frac \pi 4$-fragment. 
Notice that equation (\ref{eq:single}) holds 
 if and only if $\interp{D_1'}=\interp{D_2'}$, since $\left(\rz{}, \rz{$\pi$}\right)$ forms a basis. Thus, a variable of multiplicity $1$ can easily be removed, leading to an equivalent equation in the complete $\frac \pi 4$-fragment of the ZX-calculus.

%
 
 When a variable has a  multiplicity $r>1$ in an equation, the variable cannot be removed similarly as $\left(\rz{$\alpha$}\right)^{\otimes r}$ does not generate a basis of the $2^r$ dimensional space when $r>1$. However these dots can be replaced by an appropriate projector on the subspace generated by these dots, as described in the following.

 
%
%


%
%
%
%
%
%
 
%
%
%
 \subsubsection{When multiplicity is 2}


%
%
%


%

Consider the following diagram $R$:
\def\varone{$R$}
\[
\InputIfFileExists{2-in-2-out-box-var.tikz}{}{\input{./figures/2-in-2-out-box-var.tikz}}
~~:=~~
\InputIfFileExists{matrix-X.tikz}{}{\input{./figures/matrix-X.tikz}}
\]
One can check that $\interp R=\begin{pmatrix}
1&0&0&0\\
0&\frac{1}{2}&\frac{1}{2}&0\\
0&\frac{1}{2}&\frac{1}{2}&0\\
0&0&0&1
\end{pmatrix}$. This matrix basically mixes the second and third elements of any size-4 vector. We can then show:

\begin{lemma} For any $\alpha \in \mathbb R$, $\zx\vdash R\circ \theta_2(\alpha) = \theta_2(\alpha)$, i.e. pictorially: 
\label{lem:alphas-on-X}
\[\forall\alpha\in\mathbb{R},~~\zx\vdash~ 
\InputIfFileExists{2-gn-alpha-to-X.tikz}{}{\input{./figures/2-gn-alpha-to-X.tikz}}
\]
\end{lemma}

The proof is given in appendix.

\begin{lemma}
\label{lem:equivalence-X}
For any two ZX-diagrams $D_1,D_2 : 2 \to n$, \\$(\forall \alpha\in \mathbb R, \interp{D_1\circ \theta_2(\alpha)} =  \interp{D_2\circ \theta_2(\alpha)}) \Leftrightarrow \interp{D_1\circ R}=\interp{D_2\circ R}$ i.e.,
$$\left(\forall\alpha\in\mathbb{R},~~ \interp{
\InputIfFileExists{2-gn-alpha-to-D_1.tikz}{}{\input{./figures/2-gn-alpha-to-D_1.tikz}}
}  = \interp{
\InputIfFileExists{2-gn-alpha-to-D_2.tikz}{}{\input{./figures/2-gn-alpha-to-D_2.tikz}}
}\right) \Leftrightarrow \interp{
\InputIfFileExists{2-gn-D_1-X.tikz}{}{\input{./figures/2-gn-D_1-X.tikz}}
} = \interp{
\InputIfFileExists{2-gn-D_2-X.tikz}{}{\input{./figures/2-gn-D_2-X.tikz}}
}$$
where $\alpha$ does not appear in $D_1$ or $D_2$.
\end{lemma}

\begin{proof}
The proof consists in showing that $\interp R$ is a projector onto $S= \mathop{\mathrm{span}} \{ \interp{\theta_2(\alpha)}~|~\alpha \in \mathbb{R}\}$. According to Lemma \ref{lem:alphas-on-X}, $\interp R$ is the identity on $S$, moreover it is easy to show that $\interp R$ is a matrix of rank $3$ and that $\interp {\theta_2(0)},\interp {\theta_2(\pi/2)}, \interp {\theta_2(\pi)}$  are three linearly independent vectors in the image of $\interp R$.
%
%
%
%
%
  \end{proof}

\subsubsection{Arbitrary multiplicity}
\label{sec:multiple-variables}
%
%
%
We now want to generalise Lemma \ref{lem:equivalence-X} to any multiplicity $r$ of $\alpha$.
It turns out that there is no obvious generalization for $r$ wires of the matrix $\interp R$ expressible using angles multiple of $\frac \pi 4$, so we will rather use the following family  $(P_r)_{r\ge 2}$ of diagrams:

\[\left\lbrace\begin{array}{l}

\InputIfFileExists{2-in-2-out-box-M_2.tikz}{}{\input{./figures/2-in-2-out-box-M_2.tikz}}
~~:=~~ 
\InputIfFileExists{matrix-M2-2-arxiv.tikz}{}{\input{./figures/matrix-M2-2-arxiv.tikz}}
\\

\InputIfFileExists{matrix-M_n-def.tikz}{}{\input{./figures/matrix-M_n-def.tikz}}
~~=~~\scalebox{0.6}{
\InputIfFileExists{matrix-M_n-form.tikz}{}{\input{./figures/matrix-M_n-form.tikz}}
}
\end{array}\right.\]


For the reader convenience, here are the interpretations of $P_2$ and $P_3$:
\[
\interp{P_2} = \begin{pmatrix}
  1 & 0 & 0 & 0 \\
  0 & 0 & 1 & 0 \\
  0 & 0 & 1 & 0 \\
  0 & 0 & 0 & 1 \\
  \end{pmatrix}
\qquad \interp{P_3} = \begin{pmatrix}
  1 & 0 & 0 & 0 & 0 & 0 & 0  \\
  0 & 0 & 0 & 1 & 0 & 0 & 0 \\
  0 & 0 & 0 & 1 & 0 & 0 & 0 \\
  0 & 0 & 0 & 0 & 0 & 1 & 0  \\
  0 & 0 & 0 & 1 & 0 & 0 & 0  \\
  0 & 0 & 0 & 0 & 0 & 1 & 0 \\
  0 & 0 & 0 & 0 & 0 & 1 & 0 \\
  0 & 0 & 0 & 0 & 0 & 0 & 1  \\
  \end{pmatrix}
\]

\begin{lemma}
\label{lem:alphas-on-M}For any $r\ge 2$ and any $\alpha \in \mathbb R$, $\zx\vdash P_r\circ \theta_r(\alpha) = \theta_r(\alpha)$ i.e.,
\[\zx\vdash~ 
\InputIfFileExists{n-gn-alpha-to-M_n.tikz}{}{\input{./figures/n-gn-alpha-to-M_n.tikz}}
\]
\end{lemma}

\begin{proof}
Notice that $\interp{P_2\circ R} = \interp R$, so by completeness of the ZX-calculus for the $\frac \pi 4$ fragment, $\zx\vdash P_2\circ R = R$, so  $\zx\vdash P_2\circ R \circ \theta_2(\alpha) = R\circ \theta_2(\alpha)$. According to Lemma \ref{lem:alphas-on-X}, it implies $\zx\vdash P_2\circ \theta_2(\alpha) = \theta_2(\alpha)$. The proof for $r>2$ is by induction on $r$.  
\end{proof}

%

\begin{lemma}  \label{lem:rank}
For any $r\ge 2$, $\interp {P_r}$ is a matrix of rank at most $r+1$. 
\end{lemma}

The proof of Lemma \ref{lem:rank} is given in appendix. 


We can now prove a similar statement as in lemma \ref{lem:equivalence-X}:

\begin{lemma}
\label{lem:equivalence-Pk}
For any $r\ge 2$ and any $D_1,D_2 : r \to n$, \\$(\forall \alpha\in \mathbb R, \interp{D_1\circ \theta_r(\alpha)} =  \interp{D_2\circ \theta_r(\alpha)}) \Leftrightarrow \interp{D_1\circ P_r}=\interp{D_2\circ P_r}$ i.e.,
\begin{align*}
\left(\forall\alpha\in\mathbb{R},~~ \interp{
\InputIfFileExists{r-gn-alpha-to-D_1.tikz}{}{\input{./figures/r-gn-alpha-to-D_1.tikz}}
} = \interp{
\InputIfFileExists{r-gn-alpha-to-D_2.tikz}{}{\input{./figures/r-gn-alpha-to-D_2.tikz}}
}\right) \Leftrightarrow \interp{
\InputIfFileExists{P-to-D1.tikz}{}{\input{./figures/P-to-D1.tikz}}
} = \interp{
\InputIfFileExists{P-to-D2.tikz}{}{\input{./figures/P-to-D2.tikz}}
}
\end{align*}
where $\alpha$ does not appear in $D_1$ nor $D_2$.
\end{lemma}

%
%
%
%
%
%

\begin{proof} The proof consists in showing that $\interp {P_r}$ is a projector onto $S_r= \mathop{\mathrm{span}} \{ \interp{\theta_r(\alpha)}~|~\alpha \in \mathbb{R}\}$. According to  Lemma \ref{lem:alphas-on-M}, $\interp{P_r}$ is the identity on $S_r$, and $\interp {P_r}$ is of rank at most $r+1$ according to Lemma \ref{lem:rank}, thus to finish the proof, it is sufficient to prove that the $r+1$ vectors $(\theta_r(\alpha^{(j)}))_{j=0\ldots r}$ are linearly independent, where $\alpha^{(j)} = j\pi/r$. 

  Let $\lambda_0,...,\lambda_r$ be scalars such that $\sum_j\lambda_j\theta_r(\alpha^{(j)})=0$.
  Notice that the $2^p$-th row (when rows are labeled from $1$ to $2^r$) of $\theta_r(\alpha^{(j)})$ is exactly $e^{ip\alpha^{(j)}}$.
  Therefore, if we look at all $2^p$-th rows of the equations, we obtain
\[ \begin{pmatrix}
1&1&\cdots&1\\
e^{i\alpha^{(0)}}&e^{i\alpha^{(1)}}&\cdots&e^{i\alpha^{(r)}}\\
\vdots&\vdots&\ddots&\vdots\\
e^{in\alpha^{(0)}}&e^{in\alpha^{(1)}}&\cdots&e^{in\alpha^{(r)}}
\end{pmatrix}\begin{pmatrix}
\lambda_0\\
\lambda_1\\
\vdots\\
\lambda_r
\end{pmatrix}=0 \]
However, the first matrix is a Vandermonde matrix, with
$e^{i\alpha^{(j)}}=e^{i\alpha^{(l)}}$ iff $j=l$, which is enough to state that
this matrix is invertible. Therefore all $\lambda^{(j)}$ are equal to $0$
and the vectors $\theta_r(\alpha^{(j)})$ are linearly independent.
\end{proof}

We are now ready to prove the main theorem in the particular case of a single variable:

\begin{proposition}\label{prop:1var}
For any $D_1(\alpha), D_2(\alpha)$ ZX-diagrams linear in $\alpha$ with constants in $\frac \pi 4 \mathbb Z$,
$$ \forall  \alpha \in \mathbb R, \interp{D_1( \alpha)}= \interp{D_2( \alpha)}      \,\,\,\,\,\Leftrightarrow\,\,\,\,\,  \forall  \alpha \in \mathbb R, \zx\vdash D_1( \alpha) = D_2( \alpha)       $$
\end{proposition}
\begin{proof}
{[$\Leftarrow$]} is a direct consequence of the soundness of the ZX-calculus. 
[$\Rightarrow$] Assume $\forall  \alpha \in \mathbb R, \interp{D_1( \alpha)}= \interp{D_2( \alpha)}$. According to Proposition \ref{prop:var2inp}, $\forall  \alpha \in \mathbb R, \interp{D'_1\circ \theta_r(\alpha)}= \interp{D'_2\circ \theta_r( \alpha)}$ where $D'_i$ are in the $\frac \pi 4$-fragment of the ZX-calculus. It implies, according to Lemma \ref{lem:equivalence-Pk}, that  $\interp{D_1'\circ P_r} = \interp{D_2'\circ P_r}$. Thanks to the completeness of the ZX-calculus for the $\frac \pi 4$-fragment, $\zx\vdash D_1'\circ P_r= D_2'\circ P_r$, so $\forall \alpha\in \mathbb R, \zx\vdash D_1'\circ P_r \circ \theta_r(\alpha)= D_2'\circ P_r\circ \theta_r(\alpha)$. Thus, by Lemma \ref{lem:alphas-on-M}, $\forall \alpha\in \mathbb R, \zx\vdash D_1' \circ \theta_r(\alpha)= D_2'\circ  \theta_r(\alpha)$, which is equivalent to $\forall  \alpha \in \mathbb R, \zx\vdash D_1( \alpha) = D_2( \alpha)$ according to Proposition \ref{prop:var2inp}. 
\end{proof}

\subsection{Multiple variables}

Proposition \ref{prop:var2inp} can be straighforwardly extended to multiple variables:

\begin{proposition}\label{prop:vars2inp}
For any $D_1(\alpha), D_2(\alpha): n\to m$ two ZX-diagrams linear in $\vec \alpha=\alpha_1,\ldots, \alpha_k$ with constants in $\frac \pi 4$$\mathbb Z$, there exist  $D_1', D'_2:(\sum_{i=1}^k r_i)\to n+m$ two ZX-diagrams with angles multiple of $\frac \pi 4$ such that, for any $\vec \alpha \in \mathbb R^k$, 
\begin{equation}\label{eq:varasinputs}D_1(\vec \alpha)= D_2(\vec \alpha) \quad\Leftrightarrow\quad D_1'\circ \theta_{\vec r}(\vec \alpha) =  D_2'\circ \theta_{\vec r}(\vec \alpha)\end{equation}
is provable using the axioms of the ZX-calculus, where $r_i$ is the multiplicity of $\alpha_i$ in $D_1(\vec \alpha)=D_2(\vec \alpha)$, $\vec r:=r_1, \ldots, r_k$, and $\theta_{\vec{r}}(\vec{\alpha}):=\theta_{r_1}(\alpha_1)  \otimes \ldots\otimes \theta_{r_k}(\alpha_k)$. 

Pictorially: 
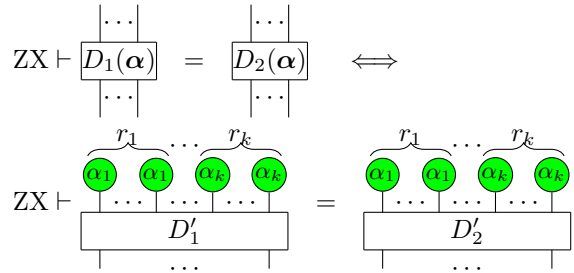
\def\fig{theta_r-alpha-on-diagrams-gen-arxiv}
\begin{align*}
&\zx\vdash\eq{}\equi{}\\
&\zx\vdash\eq{}
\end{align*}
\end{proposition}

 Similarly Lemma \ref{lem:equivalence-Pk} can also be extended to multiple variables:
\begin{lemma}
\label{lem:equivalence-Pk-vars}
For any $k\ge  0$, any $\vec r=r_1,\ldots, r_k\in \mathbb N^k$ and any $D_1,D_2 : (\sum_ir_i) \to n$, \\$(\forall \vec \alpha\in \mathbb R^k\!, \interp{D_1\circ \theta_{\vec r}(\vec \alpha)} =  \interp{D_2\circ \theta_{\vec r}(\vec \alpha)}) \Leftrightarrow \interp{D_1\circ P_{\vec r}}=\interp{D_2\circ P_{\vec r}}$
where no $\alpha_i$ appear in $D_1$ or $D_2$, and $P_{r_1,\ldots, r_k} = P_{r_1} \otimes \ldots \otimes P_{r_k}$.  
\end{lemma}

Using Proposition \ref{prop:vars2inp} and Lemma \ref{lem:equivalence-Pk-vars}, the proof of Theorem \ref{thm:provability} is a straightforward generalization of the single variable case (Proposition \ref{prop:1var}).

Notice that Theorem \ref{thm:provability} implies that if $\forall \vec \alpha \in \mathbb R^k,  \interp{D_1(\vec \alpha)}= \interp{D_2(\vec \alpha)} $ then $D_1(\vec \alpha)=D_2(\vec \alpha)$ has a \emph{uniform} proof in the ZX-calculus in the sense that the structure of the proof is the same for all the values of $\vec \alpha \in \mathbb R^k$. Indeed, following the proof of Theorem \ref{thm:provability}, the sequence of axioms which leads to a proof of $D_1(\vec \alpha)=D_2(\vec \alpha)$ is independent of the particular values of $\vec \alpha$.  Notice, however, that Theorem \ref{thm:provability} is non constructive. 

\section{Finite case-based reasoning}
\label{sec:finite-case-based-reasoning}

In order to prove that $\forall \vec \alpha \in \mathbb R^k, \zx\vdash D_1(\vec \alpha)=D_2(\vec \alpha)$ using  Theorem \ref{thm:provability}, one has to double check the semantic condition $\interp{D_1(\vec \alpha)}= \interp{D_2(\vec \alpha)}$ for all $\vec \alpha \in \mathbb R^k$, which might not be easy in practice. We show in the following two alternative ways to prove $\forall \vec \alpha \in \mathbb R^k, \zx\vdash D_1(\vec \alpha)=D_2(\vec \alpha)$ based on a finite case-based reasoning in the ZX-calculus. 

\subsection{Considering a basis}

\begin{theorem}\label{thm:basis}
For any ZX-diagrams $D_1(\vec \alpha),D_2(\vec \alpha):1\to m$ linear in $\vec \alpha=\alpha_1, \ldots, \alpha_k$ with constants in $\frac \pi 4 \mathbb Z$, if 
 $$\forall j\in \{0,1\},\forall \vec \alpha \in \mathbb R^k,  \zx\vdash {D_1(\vec \alpha)\circ R_X(j\pi)}= {D_2(\vec \alpha)\circ R_X(j\pi)}$$ 
 then $$\forall \vec \alpha \in \mathbb R^k, \zx\vdash D_1(\vec \alpha) = D_2(\vec \alpha)$$
\end{theorem}

\begin{proof}
Assume $\zx\vdash {D_1(\vec \alpha)\circ R_X(j\pi)}= {D_2(\vec \alpha)\circ R_X(j\pi)}$ for any $j\in \{0,1\}$ and any $\vec \alpha \in \mathbb R^k$. It implies that for $x\in \left\lbrace \left(\begin{array}{c}1\\0\end{array}\right),\left(\begin{array}{c}0\\1\end{array}\right)\right\rbrace$, $\interp{D_1(\vec \alpha)} x= \interp{D_2(\vec \alpha)}x$, so $\interp{D_1(\vec \alpha)} = \interp{D_2(\vec \alpha)}$, which implies according to Theorem  \ref{thm:provability} $\forall \vec \alpha \in \mathbb R^k, \zx\vdash D_1(\vec \alpha) = D_2(\vec \alpha)$. 
\end{proof}

Notice that the Theorem \ref{thm:basis} can be applied recursively: in order to prove the equality between two diagrams  with  $n$ inputs, $m$ outputs, and constants in $\frac \pi 4\mathbb Z$, one can consider the $2^{n+m}$ ways to fix these inputs/outputs in a standard  basis states. It reduces the existence of a proof between two diagrams with constants in $\frac \pi 4\mathbb Z$ to the existence of proofs on scalar diagrams (diagrams with no input and no output).   

\begin{corollary}
\label{cor:distribution}
\[\forall \alpha, \beta\in \mathbb R, \zx\vdash~~ 
\InputIfFileExists{add-axiom-2.tikz}{}{\input{./figures/add-axiom-2.tikz}}
\]
\end{corollary}
\begin{proof}
We can prove that this equality is derivable by plugging our basis $\left(\rx{},\rx{$\pi$}\right)$ on the input and one of the outputs. The detail is given in the appendix at Section \ref{prf:distribution}.
\end{proof}

\subsection{Considering a finite set of angles}


\begin{theorem}
\label{thm:valuations}
For any ZX-diagrams $D_1(\vec \alpha),D_2(\vec \alpha):n\to m$ linear in $\vec \alpha = \alpha_1, \ldots, \alpha_k$ with constants in $\frac \pi 4 \mathbb Z$, if 
 $$\forall \vec \alpha \in T_1\times \ldots \times T_k,  \zx\vdash {D_1(\vec \alpha)}= {D_2(\vec \alpha)}$$ 
 then $$\forall \vec \alpha \in \mathbb R^k, \zx\vdash D_1(\vec \alpha) = D_2(\vec \alpha)$$
 with $T_i$  a set of $\mu_i+1$ 
 distinct angles in $\mathbb{R}/2\pi\mathbb{Z}$
  where $\mu_i$ is the multiplicity of $\alpha_i$ in $D_1(\vec \alpha) = D_2(\vec \alpha)$. 
\end{theorem}

\begin{proof}
In the proof of Lemma \ref{lem:equivalence-Pk}, we actually only used $\mu_{\alpha}+1$ values of $\alpha$, that constitute a basis of $S_{\mu_{\alpha}}$. This extends naturally to several variables: the dimension of $S_{\mu_{\alpha_1}}\times\cdots\times S_{\mu_{\alpha_k}}$ is $(\mu_{\alpha_1}+1)\times\cdots\times(\mu_{\alpha_k}+1)$, and taking $\vec \alpha \in T_1\times \ldots \times T_k$ gives as many linearly independent vectors in (hence a basis of) $S_{\mu_{\alpha_1}}\times\cdots\times S_{\mu_{\alpha_k}}$.
\end{proof}

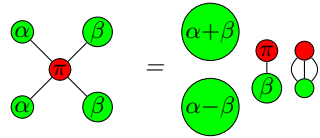
\begin{corollary}
\label{cor:big-scalar-equation}
\[
\InputIfFileExists{big-scalar-equation.tikz}{}{\input{./figures/big-scalar-equation.tikz}}
\] 
\end{corollary}
\begin{proof}
Notice that $\mu_{\alpha}=2$ in this equation. Hence we just need to evaluate it for three values of $\alpha$, for instance $0$, $\pi$ and $\frac{\pi}{2}$. We actually do not need to also evaluate $\beta$, although if we had to, since $\mu_{\beta} = 3$, we would have needed 4 different values for this variable, and so $12$ valuations for the pair $(\alpha,\beta)$.
Details are in appendix at Section \ref{prf:big-scalar-equation}.
\end{proof}

\begin{remark}
The number of occurrences of a variable is not to be mistaken for its multiplicity. For instance consider the following equation:
\[\rz{$\alpha$}=\rz{-$\alpha$}\]
This equation is obviously wrong in general, but not for $0$ and $\pi$. If we tried to apply Theorem \ref{thm:valuations} with the number of occurrences (which seems to be $1$), then we might end up with the wrong conclusion. The multiplicity (here $\mu_{\alpha}=2$) prevents this.
\end{remark}

%

\section{Diagram substitution}
\label{sec:diagram-substitution}

\begin{definition}
A diagram $D:0 \to n$ is symmetric if for any permutation $\tau$ on $\{1,\ldots n\}$, $$Q_\tau(\interp {D}) = \interp{D}$$ where $Q_\tau:\mathbb C^{2^r}\to \mathbb C^{2^r}$ is the unique morphism such that:\\
$\forall \varphi_1,\ldots, \varphi_r\in \mathbb C^{2}$, $Q_\tau(\varphi_1 \otimes \ldots \otimes \varphi_r)=\varphi_{\tau(1)} \otimes \ldots \otimes \varphi_{\tau(r)}$. 
\end{definition}

In particular for any diagram $D_0:0\to 1$, $D_0\otimes \ldots \otimes D_0$ is a symmetric diagram. 

\begin{theorem}
\label{thm:diagram-substitution}
For any $D_1(\vec \alpha), D_2(\vec \alpha):r\to n$ and any symmetric $D(\vec \alpha):0\to r$ such that  $D_1(\vec \alpha)$, $D_2(\vec \alpha)$, and $D(\vec \alpha)$ are linear in $\vec \alpha$ with constants in $\frac \pi 4\mathbb Z$, if $\forall \alpha_0\in \mathbb R, \forall\vec \alpha\in\mathbb{R}^k, \zx\vdash D_1(\vec \alpha) \circ \theta_r(\alpha_0) = D_2(\vec \alpha) \circ \theta_r(\alpha_0)$ then $\forall\vec \alpha\in\mathbb{R}^k, \zx\vdash D_1(\vec \alpha) \circ D(\vec \alpha) = D_2(\vec \alpha) \circ D(\vec  \alpha)$ i.e., pictorially:
\begin{align*}
\forall \alpha_0\in \mathbb R, \forall\vec \alpha\in\mathbb{R}^k,~~ \zx\vdash {
\InputIfFileExists{2-gn-alpha-to-D_1-bis.tikz}{}{\input{./figures/2-gn-alpha-to-D_1-bis.tikz}}
}  = {
\InputIfFileExists{2-gn-alpha-to-D_2-bis.tikz}{}{\input{./figures/2-gn-alpha-to-D_2-bis.tikz}}
}\\
 \Rightarrow ~~\forall\vec \alpha\in\mathbb{R}^k, \zx\vdash {
\InputIfFileExists{2-gn-D_1-X-bis.tikz}{}{\input{./figures/2-gn-D_1-X-bis.tikz}}
} = {
\InputIfFileExists{2-gn-D_2-X-bis.tikz}{}{\input{./figures/2-gn-D_2-X-bis.tikz}}
}
\end{align*}
\end{theorem}

\begin{proof}
If $\forall \alpha_0\in \mathbb R, \forall\vec \alpha\in\mathbb{R}^k, \zx\vdash D_1(\vec \alpha) \circ \theta_r(\alpha_0) = D_2(\vec \alpha) \circ \theta_r(\alpha_0)$ then $\interp{D_1(\vec \alpha) \circ \theta_r(\alpha_0)} = \interp{D_2(\vec \alpha) \circ \theta_r(\alpha_0)}$, so according to Lemma \ref{lem:equivalence-Pk}, $\interp{D_1(\vec \alpha) \circ P_r} = \interp{D_2(\vec \alpha) \circ P_r}$. It implies that $\zx\vdash D_1(\vec \alpha) \circ P_r  =D_2(\vec \alpha) \circ P_r$, so $\zx\vdash D_1(\vec \alpha) \circ P_r \circ D(\vec \alpha)  =D_2(\vec \alpha) \circ P_r\circ D(\vec \alpha)$. To complete the proof, it is enough to show that $\zx\vdash P_r\circ D(\vec \alpha) = D(\vec \alpha)$. \\
Let $\mathcal S = \{\interp D ~|~ÊD:0\to n \text{ symmetrical}\}$. First we show that $\mathcal S$ is of dimension at most $r+1$. 
Indeed, notice that if $\varphi \in \mathcal S$, then $\forall i,j\in \{0,\ldots , 2^r-1\}$ s.t. $|i|_1 = |j|_1$, $\varphi_i = \varphi_j$, where $|x|_1$ is the Hamming weight of the binary representation of $x$. As a consequence, for any $\varphi\in \mathcal S$, $\exists a_0, \ldots ,a_r\in \mathbb C$ s.t. $\varphi = \sum_{h=0}^na_h\varphi^{(h)}$ where $\varphi^{(h)}\in \mathbb C^{2^r}$ is defined as $\varphi^{(h)}_i=\begin{cases}1&\text{if $|i|_1 = h$}\\0&\text{otherwise}\end{cases}$. Thus $\mathcal S$ is of dimension at most $r+1$.  Moreover, for any $\alpha\in \mathbb R$, $\interp{\theta_r(\alpha)}\in \mathcal S$, so $\mathcal S \subseteq \mathcal S_r:= \mathop{\mathrm{span}} \{ \interp{\theta_r(\alpha)}~|~\alpha \in \mathbb{R}\}$. Since $\mathcal S_r$ is of dimension $r+1$ (see proof of Lemma \ref{lem:equivalence-Pk}), $\mathcal S=\mathcal S_r$. 
As a consequence $\interp D\in \mathcal S_r$, so $\interp {Pr}\circ \interp {D(\vec \alpha)} = \interp {D(\vec \alpha)}$, since, according to Lemma \ref{lem:alphas-on-M} for any $\alpha\in \mathbb R$, $\interp{P_r\circ \theta_r(\alpha)}= \interp {\theta_r(\alpha)}$. Thus, $\zx\vdash P_r\circ D(\vec \alpha)$ thanks to Theorem \ref{thm:provability}.
\end{proof}

\begin{corollary}
\label{cor:gen-supp}
 \[\forall\alpha,\beta\in\mathbb{R}^2,\quad \zx\vdash
\InputIfFileExists{new-supplementarity.tikz}{}{\input{./figures/new-supplementarity.tikz}}
\]
\end{corollary}
\begin{proof}
 Indeed, simply by decomposing the colour-swapped version of \supp{} using \so{}, we can derive:
 \[\forall\alpha\in\mathbb{R},\quad\zx\vdash 
\InputIfFileExists{decomposed-supp.tikz}{}{\input{./figures/decomposed-supp.tikz}}
\]
Now we just need to apply Theorem \ref{thm:diagram-substitution} with 
$$
\InputIfFileExists{diag-substitution-gen-supp-arxiv.tikz}{}{\input{./figures/diag-substitution-gen-supp-arxiv.tikz}}
$$ which is clearly symmetrical, and use \so{} to merge the adjacent red nodes.
\end{proof}

\section{Completion of ZX-calculus for general quantum mechanics}
\label{sec:completion}

\subsection{Incompleteness}

\begin{figure*}[!b]
 \centering
 \hypertarget{r:rules-c}{}
\begin{tabular}{|c@{$\qquad$}c|}
   \hline
   & \\
   \multicolumn{2}{|c|}{
\InputIfFileExists{spider-1.tikz}{}{\input{./figures/spider-1.tikz}}
 $\qquad$ 
\InputIfFileExists{s2-simple.tikz}{}{\input{./figures/s2-simple.tikz}}
 $\qquad$ 
\InputIfFileExists{induced_compact_structure-2wire.tikz}{}{\input{./figures/induced_compact_structure-2wire.tikz}}
}\\
   & \\
   
\InputIfFileExists{k2s.tikz}{}{\input{./figures/k2s.tikz}}
&
\InputIfFileExists{bicolor_pi_4_eq_empty.tikz}{}{\input{./figures/bicolor_pi_4_eq_empty.tikz}}
\\
   & \\
   
\InputIfFileExists{b1s.tikz}{}{\input{./figures/b1s.tikz}}
& 
\InputIfFileExists{b2s.tikz}{}{\input{./figures/b2s.tikz}}
\\
   & \\
   
\InputIfFileExists{euler-decomp-scalar-free.tikz}{}{\input{./figures/euler-decomp-scalar-free.tikz}}
&  
\InputIfFileExists{h2.tikz}{}{\input{./figures/h2.tikz}}
\\
   & \\
   
\InputIfFileExists{commutation-of-controls-general-simplified.tikz}{}{\input{./figures/commutation-of-controls-general-simplified.tikz}}
& 
\InputIfFileExists{former-supp.tikz}{}{\input{./figures/former-supp.tikz}}
 \\
   & \\
   
\InputIfFileExists{BW-simplified.tikz}{}{\input{./figures/BW-simplified.tikz}}
 & 
\InputIfFileExists{add-axiom-3.tikz}{}{\input{./figures/add-axiom-3.tikz}}
\\
   & \\
   \hline
  \end{tabular}
 \caption[]{Set of rules for the general \zxcalc with scalars, denoted $\zx_c$. All of these rules also hold when flipped upside-down, or with the colours red and green swapped. The right-hand side of (E) is an empty diagram. (...) denote zero or more wires, while (\protect\rotatebox{45}{\raisebox{-0.5em}{$\cdots$}}) denote one or more wires.}
 \label{fig:ZX_rules-complete}
\end{figure*}

The axiomatisation of ZX-calculus (figure \ref{fig:ZX_rules}) is complete for the Clifford+T quantum mechanics  --i.e. the $\frac \pi 4$-fragment--, but is not complete in general:

\begin{theorem}
\label{thm:incompleteness}
There exist two ZX-diagrams $D_1$ and $D_2$ such that:
\[\interp{D_1}=\interp{D_2}\qquad\text{and}\qquad \zx\nvdash D_1=D_2\]
\end{theorem}
\begin{proof}
Consider the following equation:
\[
\InputIfFileExists{incomplete-pi_3.tikz}{}{\input{./figures/incomplete-pi_3.tikz}}
\]
This equation is sound, it represents $$(1+e^{i\frac{2\pi}{3}})(1+e^{i\frac{4\pi}{3}})=1+e^{i\frac{2\pi}{3}}+e^{i\frac{4\pi}{3}}+e^{i\frac{6\pi}{3}}=1$$
However, consider the interpretation $\interp{.}_9$ that multiplies all the angles by $9$. All the multiples of $\frac{\pi}{4}$ remain unchanged ($\frac{k\pi}{4}\times9=\frac{k\pi}{4}+2k\pi=\frac{k\pi}{4}$). It is then easy to show that all the rules in Figure \ref{fig:ZX_rules} hold with this interpretation. However:
\def\fig{incomplete-pi_3-proof}
\[\interp{}_9\eq{}\ \neq\ \eq{}\interp{}_9\]
Indeed the left hand side amounts to $4$ while the right hand side amounts to $1$.
Since all the rules in Figure \ref{fig:ZX_rules} hold with this interpretation, if the calculus were complete, then it would prove the above equation and so its interpretation would hold. It does not, so the \zxcalc is not complete.
\end{proof}

Notice that thanks to Theorem \ref{thm:provability}, a completion of the ZX calculus would imply to  add  either non linear axioms, or axioms with constants not multiple of $\pi/4$. Such potential axioms have already been discovered, for instance the cyclotomic supplementarity \cite{cyclo}:
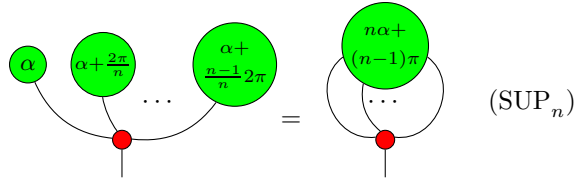
\[
\InputIfFileExists{cyclo-supp.tikz}{}{\input{./figures/cyclo-supp.tikz}}
\quad\text{(SUP}_n\text{)}\]
Adding this family of axioms to those of Figure \ref{fig:ZX_rules} would nullify the counterexample in the proof of \ref{thm:incompleteness} (the equality is derivable from ZX+(SUP$_3$)). However, the \zxcalc, with this set of axioms, would still be incomplete. Indeed, the argument given in \cite{cyclo} still holds here.

In the following, we actually show that adding one axiom to the set in Figure \ref{fig:ZX_rules} is sufficient to get the completeness in general. Contrary to the previous family of axioms, this one manipulates angles in a non-linear fashion.

\subsection{A complete axiomatisation}

We add a new axiom \add{-c} to the previous set of axioms, and define $\zx_c$ as the resulting set of axioms. This set is given in Figure \ref{fig:ZX_rules-complete}. The side condition $2e^{i\theta_3}\cos(\gamma)=e^{i\theta_1}\cos(\alpha)+e^{i\theta_2}\cos(\beta)$ forces this axiom to be non-linear. As announced:

\begin{theorem}
\label{thm:completeness}
The set of rules $\zx_c$ (Figure \ref{fig:ZX_rules-complete}) is complete. For any two ZX-diagrams $D_1$ and $D_2$:
\[\interp{D_1}=\interp{D_2}\equi{}\zx_c\vdash D_1=D_2\]
\end{theorem}

The rest of the article is dedicated to the proof of this theorem.

\subsubsection{ZW-Calculus}
\label{sec:zw}

To do so, as in \cite{JPV,NgWang}, we will use the completeness of another graphical calculus for quantum mechanics called ZW-Calculus, that we present in this section.


The GHZ/W-Calculus, developed by Coecke and Kissinger \cite{ghz-w}, has been turned into another language, called ZW-Calculus by Hazihasanovic, who also proved its completeness \cite{zw}. This language initially dealt with matrices over $\mathbb{Z}$, but it has been expanded later on, and its more universal version deals with $\mathbb{C}$ \cite{Amar}. It is generated by:
\begin{align*}
T_e=\left\lbrace 

\InputIfFileExists{Z-param.tikz}{}{\input{./figures/Z-param.tikz}}
~,~
\begin{tikzpicture}
	\begin{pgfonlayer}{nodelayer}
		\node [style=dot] (0) at (0, -0) {};
		\node [style=none] (1) at (0, 0.5) {};
		\node [style=none] (2) at (0, -0.5) {};
	\end{pgfonlayer}
	\begin{pgfonlayer}{edgelayer}
		\draw (1.center) to (2.center);
	\end{pgfonlayer}
\end{tikzpicture}}
~,~
\begin{tikzpicture}
	\begin{pgfonlayer}{nodelayer}
		\node [style=dot] (0) at (0, 0) {};
		\node [style=none] (1) at (0, 0.5000001) {};
		\node [style=none] (2) at (-0.25, -0.5) {};
		\node [style=none] (3) at (0.25, -0.5) {};
	\end{pgfonlayer}
	\begin{pgfonlayer}{edgelayer}
		\draw (0) to (2.center);
		\draw (3.center) to (0);
		\draw (0) to (1.center);
	\end{pgfonlayer}
\end{tikzpicture}}
~,\scalebox{1}{
\begin{tikzpicture}
	\begin{pgfonlayer}{nodelayer}
		\node [style=none] (0) at (0, 0.2499999) {};
		\node [style=none] (1) at (0, -0.2499999) {};
	\end{pgfonlayer}
	\begin{pgfonlayer}{edgelayer}
		\draw (0.center) to (1.center);
	\end{pgfonlayer}
\end{tikzpicture}}
}~,\raisebox{-0.3em}{
\begin{tikzpicture}
	\begin{pgfonlayer}{nodelayer}
		\node [style=none] (0) at (-0.2500001, 0.2500001) {};
		\node [style=none] (1) at (0.2500001, 0.2500001) {};
	\end{pgfonlayer}
	\begin{pgfonlayer}{edgelayer}
		\draw [bend right=90, looseness=1.75] (0.center) to (1.center);
	\end{pgfonlayer}
\end{tikzpicture}}
}~,\raisebox{-0.4em}{
\begin{tikzpicture}
	\begin{pgfonlayer}{nodelayer}
		\node [style=none] (0) at (-0.2500001, -0) {};
		\node [style=none] (1) at (0.2500001, -0) {};
	\end{pgfonlayer}
	\begin{pgfonlayer}{edgelayer}
		\draw [bend left=90, looseness=1.75] (0.center) to (1.center);
	\end{pgfonlayer}
\end{tikzpicture}}
}~,~
\InputIfFileExists{crossing.tikz}{}{\input{./figures/crossing.tikz}}
~,~
\InputIfFileExists{zw-cross.tikz}{}{\input{./figures/zw-cross.tikz}}
~,~
\InputIfFileExists{empty-diagram.tikz}{}{\input{./figures/empty-diagram.tikz}}
~
\right\rbrace\\
\tag*{$\substack{n,m\in\mathbb{N}, r\in \mathbb{C}}$}
\end{align*}
and diagrams are created thanks to the two same -- spacial and sequential -- compositions.

The diagrams represent matrices, in accordance to the standard interpretation, that associates to any diagram of the ZW-Calculus $D$ with $n$ inputs and $m$ outputs, a linear map $\interp{D}:\mathbb{C}^{2^n}\to\mathbb{C}^{2^m}$, inductively defined as:\\
\titlerule{$\interp{.}$}
$$ \interp{D_1\otimes D_2}:=\interp{D_1}\otimes\interp{D_2} \qquad 
\interp{D_2\circ D_1}:=\interp{D_2}\circ\interp{D_1}$$
$$\interp{
\InputIfFileExists{empty-diagram.tikz}{}{\input{./figures/empty-diagram.tikz}}
~}:=\begin{pmatrix}
1
\end{pmatrix} \qquad
\interp{~
\begin{tikzpicture}
	\begin{pgfonlayer}{nodelayer}
		\node [style=none] (0) at (0, 0.2499999) {};
		\node [style=none] (1) at (0, -0.2499999) {};
	\end{pgfonlayer}
	\begin{pgfonlayer}{edgelayer}
		\draw (0.center) to (1.center);
	\end{pgfonlayer}
\end{tikzpicture}}
~~}:= \begin{pmatrix}
1 & 0 \\ 0 & 1\end{pmatrix} \qquad
\interp{\raisebox{-0.3em}{$
\begin{tikzpicture}
	\begin{pgfonlayer}{nodelayer}
		\node [style=none] (0) at (-0.2500001, 0.2500001) {};
		\node [style=none] (1) at (0.2500001, 0.2500001) {};
	\end{pgfonlayer}
	\begin{pgfonlayer}{edgelayer}
		\draw [bend right=90, looseness=1.75] (0.center) to (1.center);
	\end{pgfonlayer}
\end{tikzpicture}}
$}}:= \begin{pmatrix}
1&0&0&1
\end{pmatrix}$$
$$\interp{
\InputIfFileExists{crossing.tikz}{}{\input{./figures/crossing.tikz}}
}:= \begin{pmatrix}
1&0&0&0\\
0&0&1&0\\
0&1&0&0\\
0&0&0&1
\end{pmatrix} \qquad
\interp{
\InputIfFileExists{zw-cross.tikz}{}{\input{./figures/zw-cross.tikz}}
}:= \begin{pmatrix}
1&0&0&0\\
0&0&1&0\\
0&1&0&0\\
0&0&0&-1
\end{pmatrix}$$
$$ \interp{\raisebox{-0.4em}{$
\begin{tikzpicture}
	\begin{pgfonlayer}{nodelayer}
		\node [style=none] (0) at (-0.2500001, -0) {};
		\node [style=none] (1) at (0.2500001, -0) {};
	\end{pgfonlayer}
	\begin{pgfonlayer}{edgelayer}
		\draw [bend left=90, looseness=1.75] (0.center) to (1.center);
	\end{pgfonlayer}
\end{tikzpicture}}
$}}:= \begin{pmatrix}
1\\0\\0\\1
\end{pmatrix}\qquad\interp{
\begin{tikzpicture}
	\begin{pgfonlayer}{nodelayer}
		\node [style=dot] (0) at (0, -0) {};
		\node [style=none] (1) at (0, 0.5) {};
		\node [style=none] (2) at (0, -0.5) {};
	\end{pgfonlayer}
	\begin{pgfonlayer}{edgelayer}
		\draw (1.center) to (2.center);
	\end{pgfonlayer}
\end{tikzpicture}}
}:= \begin{pmatrix}
0&1\\1&0
\end{pmatrix} \qquad 
\interp{
\begin{tikzpicture}
	\begin{pgfonlayer}{nodelayer}
		\node [style=dot] (0) at (0, 0) {};
		\node [style=none] (1) at (0, 0.5000001) {};
		\node [style=none] (2) at (-0.25, -0.5) {};
		\node [style=none] (3) at (0.25, -0.5) {};
	\end{pgfonlayer}
	\begin{pgfonlayer}{edgelayer}
		\draw (0) to (2.center);
		\draw (3.center) to (0);
		\draw (0) to (1.center);
	\end{pgfonlayer}
\end{tikzpicture}}
}:= \begin{pmatrix}
0&1\\1&0\\1&0\\0&0
\end{pmatrix}$$
$$
\interp{\begin{tikzpicture}
	\begin{pgfonlayer}{nodelayer}
		\node [style=white dot] (0) at (0, -0) {};
		\node [style=none] (1) at (-0.25, -0) {$r$};
	\end{pgfonlayer}
\end{tikzpicture}}=\begin{pmatrix}1+r\end{pmatrix}\qquad\qquad
\begin{array}{r}
\interp{
\InputIfFileExists{Z-param.tikz}{}{\input{./figures/Z-param.tikz}}
}=\annoted{2^m}{2^n}{\begin{pmatrix}
  1 & 0 & \cdots & 0 & 0 \\
  0 & 0 & \cdots & 0 & 0 \\
  \vdots & \vdots & \ddots & \vdots & \vdots \\
  0 & 0 & \cdots & 0 & 0 \\
  0 & 0 & \cdots & 0 & r
 \end{pmatrix}}\\
 (n+m>0)
\end{array}
$$
\rule{\columnwidth}{0.5pt}
We use the same notation for the two standard interpretations, from either one of the two languages to their corresponding  matrices.

When a white dot has no visible parameter, then $1$ is implicitly used.


The ZW-Calculus comes with its own set of axioms, depicted in appendix in Section \ref{sec:rules-zw}.
The paradigm ``only topology matters'' still stands here, and gives a number of implicit rules, the same way it does with the ZX-Calculus, but for one node, \scalebox{0.7}{
\InputIfFileExists{zw-cross.tikz}{}{\input{./figures/zw-cross.tikz}}
}, for which the order of inputs and outputs matters. Here again, one can transform a diagram into an equivalent one by locally applying the axioms of the ZW: For any three diagrams of the ZW-Calculus, $D_1, D_2$, and $D$, if $\zw\vdash D_1=D_2$, then:\\
\renewcommand*{\arraystretch}{1.2}
\begin{tabular}{@{$\qquad$}l@{$\qquad$}l}
$\bullet\zw\vdash D_1\circ D = D_2\circ D$ & $\bullet \zw\vdash D\circ D_1 = D\circ D_2$\\
$\bullet\zw\vdash D_1\otimes D = D_2\otimes D$ & $\bullet \zw\vdash D\otimes D_1 = D\otimes D_2$
\end{tabular}\\
\renewcommand*{\arraystretch}{0.8}


\subsubsection{Interpretations from ZX to ZW and back}
Both the \zxcalc and the ZW-Calculus are universal for complex matrices, so there exists a pair of translations between the two languages which preserve the semantics ($[.]_X: ZW\to ZX$ and $[.]_W:ZX\to ZW$ s.t. $\forall D\in ZX, \interp{[D]_W}=\interp D$ and $\forall D\in ZW, \interp{[D]_X}=\interp D$). The axiom \add{-c} has been chosen so that we can prove that $\zx\vdash {[[D]_W]_X}= D$ for any generator $D$ of the ZX-calculus and that $\zx\vdash [D_1]_X=[D_2]_X$ for any axiom $D_1=D_2$ of the ZW calculus. The choice of the translations is however essential as the new axiom relies on them. 

The $[.]_W$ translation 
can be canonically defined using the normal form of the ZW-calculus: for any generator $D$ of the ZX one can define $[D]_W$ as the ZW normal form representation of the matrix $\interp{D}$. It is however convenient to deviate from this canonically defined interpretation for the green and red spiders and for the Hadamard gate. We end up with basically the same translation from ZX to ZW as in \cite{NgWang}:\\

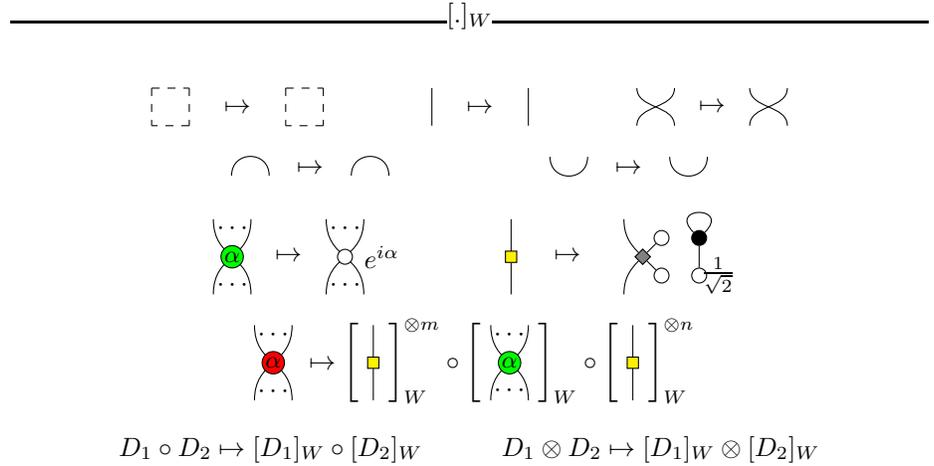
\titlerule{$[.]_W$}\\
$$
\InputIfFileExists{empty-diagram.tikz}{}{\input{./figures/empty-diagram.tikz}}
 \quad\mapsto\quad 
\InputIfFileExists{empty-diagram.tikz}{}{\input{./figures/empty-diagram.tikz}}
\qquad\qquad

\begin{tikzpicture}
	\begin{pgfonlayer}{nodelayer}
		\node [style=none] (0) at (0, 0.2499999) {};
		\node [style=none] (1) at (0, -0.2499999) {};
	\end{pgfonlayer}
	\begin{pgfonlayer}{edgelayer}
		\draw (0.center) to (1.center);
	\end{pgfonlayer}
\end{tikzpicture}}
 \quad\mapsto\quad 
\begin{tikzpicture}
	\begin{pgfonlayer}{nodelayer}
		\node [style=none] (0) at (0, 0.2499999) {};
		\node [style=none] (1) at (0, -0.2499999) {};
	\end{pgfonlayer}
	\begin{pgfonlayer}{edgelayer}
		\draw (0.center) to (1.center);
	\end{pgfonlayer}
\end{tikzpicture}}
\qquad\qquad
\InputIfFileExists{ZW-to-ZX-braid-no-braid.tikz}{}{\input{./figures/ZW-to-ZX-braid-no-braid.tikz}}
$$
$$
\begin{tikzpicture}
	\begin{pgfonlayer}{nodelayer}
		\node [style=none] (0) at (-0.2500001, -0) {};
		\node [style=none] (1) at (0.2500001, -0) {};
	\end{pgfonlayer}
	\begin{pgfonlayer}{edgelayer}
		\draw [bend left=90, looseness=1.75] (0.center) to (1.center);
	\end{pgfonlayer}
\end{tikzpicture}}
 \quad\raisebox{0.3em}{$\mapsto$}\quad 
\begin{tikzpicture}
	\begin{pgfonlayer}{nodelayer}
		\node [style=none] (0) at (-0.2500001, -0) {};
		\node [style=none] (1) at (0.2500001, -0) {};
	\end{pgfonlayer}
	\begin{pgfonlayer}{edgelayer}
		\draw [bend left=90, looseness=1.75] (0.center) to (1.center);
	\end{pgfonlayer}
\end{tikzpicture}}
\hspace{6em}

\begin{tikzpicture}
	\begin{pgfonlayer}{nodelayer}
		\node [style=none] (0) at (-0.2500001, 0.2500001) {};
		\node [style=none] (1) at (0.2500001, 0.2500001) {};
	\end{pgfonlayer}
	\begin{pgfonlayer}{edgelayer}
		\draw [bend right=90, looseness=1.75] (0.center) to (1.center);
	\end{pgfonlayer}
\end{tikzpicture}}
 \quad\raisebox{0.3em}{$\mapsto$}\quad 
\begin{tikzpicture}
	\begin{pgfonlayer}{nodelayer}
		\node [style=none] (0) at (-0.2500001, 0.2500001) {};
		\node [style=none] (1) at (0.2500001, 0.2500001) {};
	\end{pgfonlayer}
	\begin{pgfonlayer}{edgelayer}
		\draw [bend right=90, looseness=1.75] (0.center) to (1.center);
	\end{pgfonlayer}
\end{tikzpicture}}
$$
$$
\InputIfFileExists{gn-alpha-ZX-to-ZW.tikz}{}{\input{./figures/gn-alpha-ZX-to-ZW.tikz}}
\qquad\qquad

\InputIfFileExists{hadamard-ZX-to-ZW.tikz}{}{\input{./figures/hadamard-ZX-to-ZW.tikz}}
$$
$$
\InputIfFileExists{rn-alpha-2.tikz}{}{\input{./figures/rn-alpha-2.tikz}}
\mapsto\left[~
\begin{tikzpicture}
	\begin{pgfonlayer}{nodelayer}
		\node [style={H box}] (0) at (0, 0) {};
		\node [style=none] (1) at (0, 0.5) {};
		\node [style=none] (2) at (0, -0.5) {};
	\end{pgfonlayer}
	\begin{pgfonlayer}{edgelayer}
		\draw (2.center) to (1.center);
	\end{pgfonlayer}
\end{tikzpicture}
}
~\right]_W^{\otimes m}\circ \left[
\InputIfFileExists{gn-alpha-2.tikz}{}{\input{./figures/gn-alpha-2.tikz}}
\right]_W\circ \left[~
\begin{tikzpicture}
	\begin{pgfonlayer}{nodelayer}
		\node [style={H box}] (0) at (0, 0) {};
		\node [style=none] (1) at (0, 0.5) {};
		\node [style=none] (2) at (0, -0.5) {};
	\end{pgfonlayer}
	\begin{pgfonlayer}{edgelayer}
		\draw (2.center) to (1.center);
	\end{pgfonlayer}
\end{tikzpicture}
}
~\right]_W^{\otimes n}
$$
\noindent\begin{minipage}{\columnwidth}
$$D_1\circ D_2\mapsto [D_1]_W\circ[D_2]_W\qquad\quad D_1\otimes D_2\mapsto [D_1]_W\otimes[D_2]_W$$
\rule{\columnwidth}{0.5pt}
\end{minipage}\\


The $[.]_X$ translation 
has already been partially defined in \cite{JPV}. To extend it to the generalised white spider present in ZW, the main subtlety is the encoding of positive real numbers 
 in the ZX-diagrams. 
In \cite{NgWang}, the authors decompose, roughly speaking, a positive real number  into its integer part and its non-integer part. 
Our translation relies on a different (although not unique) decomposition:
\[\forall z\in\mathbb{C},~~ \exists (n,\theta,\beta)\in\mathbb{N}\times[0;2\pi[\times\left[0;\frac{\pi}{2}\right],\quad z=2^n\cos(\beta)e^{i\theta}\]

\titlerule{$[.]_X$}\\
$$
\InputIfFileExists{empty-diagram.tikz}{}{\input{./figures/empty-diagram.tikz}}
 \quad\mapsto\quad 
\InputIfFileExists{empty-diagram.tikz}{}{\input{./figures/empty-diagram.tikz}}
\qquad\qquad

\begin{tikzpicture}
	\begin{pgfonlayer}{nodelayer}
		\node [style=none] (0) at (0, 0.2499999) {};
		\node [style=none] (1) at (0, -0.2499999) {};
	\end{pgfonlayer}
	\begin{pgfonlayer}{edgelayer}
		\draw (0.center) to (1.center);
	\end{pgfonlayer}
\end{tikzpicture}}
 \quad\mapsto\quad 
\begin{tikzpicture}
	\begin{pgfonlayer}{nodelayer}
		\node [style=none] (0) at (0, 0.2499999) {};
		\node [style=none] (1) at (0, -0.2499999) {};
	\end{pgfonlayer}
	\begin{pgfonlayer}{edgelayer}
		\draw (0.center) to (1.center);
	\end{pgfonlayer}
\end{tikzpicture}}
\qquad\qquad
\InputIfFileExists{ZW-to-ZX-braid-no-braid.tikz}{}{\input{./figures/ZW-to-ZX-braid-no-braid.tikz}}
$$
$$
\begin{tikzpicture}
	\begin{pgfonlayer}{nodelayer}
		\node [style=none] (0) at (-0.2500001, -0) {};
		\node [style=none] (1) at (0.2500001, -0) {};
	\end{pgfonlayer}
	\begin{pgfonlayer}{edgelayer}
		\draw [bend left=90, looseness=1.75] (0.center) to (1.center);
	\end{pgfonlayer}
\end{tikzpicture}}
 \quad\raisebox{0.3em}{$\mapsto$}\quad 
\begin{tikzpicture}
	\begin{pgfonlayer}{nodelayer}
		\node [style=none] (0) at (-0.2500001, -0) {};
		\node [style=none] (1) at (0.2500001, -0) {};
	\end{pgfonlayer}
	\begin{pgfonlayer}{edgelayer}
		\draw [bend left=90, looseness=1.75] (0.center) to (1.center);
	\end{pgfonlayer}
\end{tikzpicture}}
\hspace{6em}

\begin{tikzpicture}
	\begin{pgfonlayer}{nodelayer}
		\node [style=none] (0) at (-0.2500001, 0.2500001) {};
		\node [style=none] (1) at (0.2500001, 0.2500001) {};
	\end{pgfonlayer}
	\begin{pgfonlayer}{edgelayer}
		\draw [bend right=90, looseness=1.75] (0.center) to (1.center);
	\end{pgfonlayer}
\end{tikzpicture}}
 \quad\raisebox{0.3em}{$\mapsto$}\quad 
\begin{tikzpicture}
	\begin{pgfonlayer}{nodelayer}
		\node [style=none] (0) at (-0.2500001, 0.2500001) {};
		\node [style=none] (1) at (0.2500001, 0.2500001) {};
	\end{pgfonlayer}
	\begin{pgfonlayer}{edgelayer}
		\draw [bend right=90, looseness=1.75] (0.center) to (1.center);
	\end{pgfonlayer}
\end{tikzpicture}}
$$
$$\begin{array}{c}

\InputIfFileExists{ZW-to-ZX-dot-1-1.tikz}{}{\input{./figures/ZW-to-ZX-dot-1-1.tikz}}
\\\\

\InputIfFileExists{ZW-to-ZX-cross-no-braid.tikz}{}{\input{./figures/ZW-to-ZX-cross-no-braid.tikz}}

\end{array} \qquad\quad

\InputIfFileExists{ZW-to-ZX-dot-1-2.tikz}{}{\input{./figures/ZW-to-ZX-dot-1-2.tikz}}

$$
$$
\InputIfFileExists{ZW-white-dot-to-ZX.tikz}{}{\input{./figures/ZW-white-dot-to-ZX.tikz}}
$$
\noindent\begin{minipage}{\columnwidth}
$$D_1\circ D_2\mapsto [D_1]_X\circ[D_2]_X\quad\qquad D_1\otimes D_2\mapsto [D_1]_X\otimes[D_2]_X$$
\rule{\columnwidth}{0.5pt}
\end{minipage}\\

\begin{remark}
$n$ is well-defined: Every complex number $x\neq0$ can be expressed as $\rho e^{i\theta}$ where $\rho\in\mathbb{R}^*+$. If $x=0$, then $n:=0$. However, $\theta$ may take any value, but it makes no difference (see Section \ref{prf:X-is-homomorphism} in appendix).
\end{remark}


We may prove the two following propositions:

\begin{proposition}
\label{prop:double-interpretation-equivalence-1}
\[\zx_c\vdash D = [[D]_W]_X\]
\end{proposition}
Proof in appendix at Section \ref{prf:double-interp-eq}.

\begin{proposition}
\label{prop:X-is-homomorphism}
\[\zw\vdash D_1=D_2 \implies \zx_c\vdash [D_1]_X=[D_2]_X\]
\end{proposition}
Proof in appendix at Section \ref{prf:X-is-homomorphism}.

The completeness of the calculus is now easy to prove:
\begin{proof}[Theorem \ref{thm:completeness}]
Let $D_1$ and $D_2$ be two diagrams of the ZX-Calculus such that $\interp{D_1}=\interp{D_2}$. Since $[.]_W$ preserves the the semantics, $\interp{[D_1]_W}=\interp{[D_1]_W}$. By completeness of the ZW-Calculus, $\zw\vdash [D_1]_W=[D_2]_W$. By Proposition \ref{prop:X-is-homomorphism}, $\zx_c\vdash [[D_1]_W]_X=[[D_2]_W]_X$. Finally, by Proposition \ref{prop:double-interpretation-equivalence-1}, $\zx_c\vdash D_1=D_2$ which completes the proof.
\end{proof}

\section{Discussion}
\label{sec:discussion}

Together with the 12 axioms used for the Clifford+T completeness, the present complete axiomatisation is composed of 13 axioms, i.e. (less than) half of the 32 axioms in \cite{NgWang}. Moreover our axiomatisation is ``retro-compatible" in the sense that any proof being derived so far with some previous version of the ZX-calculus can be straightforwardly derived using this set of axioms. Indeed, this set of axioms has been obtained after successive refinements of the original axiomatisation of the ZX-calculus, where every discarded axiom has been constructively proved to be derivable using the remaining axioms. 

The rule \add{-c} comes with a side condition on the affected angles: $2e^{i\theta_3}\cos(\gamma)=e^{i\theta_1}\cos(\alpha)$ $+e^{i\theta_2}\cos(\beta)$.  In order to claim that the ZX-calculus is complete without the help of some external computations, axiom \add{-c} must be seen as an infinite (uncountable) family of axioms. Notice that other axioms (e.g. \so{-c}, \kt{-c}) also involve some operations ($\alpha+\beta$ or $-\alpha$) however these Phase group operations are not side operations, but on the contrary fundamental properties on which the ZX-calculus has been built. 
In this sense the complete axiomatisation is ``pseudo-finite", and the quest for a complete and finite axiomatisation of the ZX-calculus for a non-approximative universal fragment is still open. One way to achieve such finite completeness would be to provide translations $[.]_X$ and $[.]_W$ between the ZX and ZW calculi which somehow \emph{preserve} the phase group structure of the ZX-Calculus and the ring structure of the ZW-Calculus. Notice however that \cite{incompleteness} and \cite{cyclo} are two different kinds of evidence that such a finite complete axiomatisation may not exist.




%
\section*{Acknowledgements}
The authors acknowledge support from the projects ANR-17-CE25-0009 SoftQPro, ANR-17-CE24-0035 VanQuTe, PIA-GDN/Quantex, and STIC-AmSud 16-STIC-05 FoQCoSS. All diagrams were written with the help of TikZit.

\appendix
\section{Appendix}

As in \cite{JPV}, we define the triangle as a syntactic sugar for a bigger diagram:
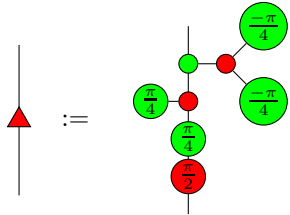
\[
\InputIfFileExists{ug-decomp.tikz}{}{\input{./figures/ug-decomp.tikz}}
\]
It has interpretation $\interp{
\begin{tikzpicture}
	\begin{pgfonlayer}{nodelayer}
		\node [style=ug] (0) at (0, -0) {};
		\node [style=none] (1) at (0, 0.5000001) {};
		\node [style=none] (2) at (0, -0.5000001) {};
		\node [style=none] (3) at (0, -0.7499999) {};
		\node [style=none] (4) at (0, 0.7500001) {};
	\end{pgfonlayer}
	\begin{pgfonlayer}{edgelayer}
		\draw (1.center) to (2.center);
	\end{pgfonlayer}
\end{tikzpicture}}
} = \begin{pmatrix}1&1\\0&1\end{pmatrix}$.

\subsection{Lemmas}

We first give a few useful lemmas:

\noindent\begin{minipage}{\columnwidth}
\begin{multicols}{2}
\begin{lemma}
\label{lem:2-is-sqrt2-squared}
\[
\InputIfFileExists{2-is-sqrt2-squared.tikz}{}{\input{./figures/2-is-sqrt2-squared.tikz}}
\]
\end{lemma}
\begin{lemma}
\label{lem:hopf}
\[
\InputIfFileExists{hopf.tikz}{}{\input{./figures/hopf.tikz}}
\]
\end{lemma}
\end{multicols}
\end{minipage}\\\\

\noindent\begin{minipage}{\columnwidth}
\begin{multicols}{2}
\begin{lemma}
\label{lem:inverse}
\[
\InputIfFileExists{inverse.tikz}{}{\input{./figures/inverse.tikz}}
\]
\end{lemma}

\begin{lemma}
\label{lem:k1}
\[
\InputIfFileExists{k1.tikz}{}{\input{./figures/k1.tikz}}
\]
\end{lemma}
\end{multicols}
\end{minipage}\\\\

\noindent\begin{minipage}{\columnwidth}
\begin{multicols}{2}
\begin{lemma}
\label{lem:multiplying-global-phases}
\[
\InputIfFileExists{multiplying-global-phases.tikz}{}{\input{./figures/multiplying-global-phases.tikz}}
\]
\end{lemma}
\begin{lemma}
\label{lem:bicolor-0-alpha}
\[
\InputIfFileExists{bicolor-0-alpha.tikz}{}{\input{./figures/bicolor-0-alpha.tikz}}
\]
\end{lemma}
\end{multicols}
\end{minipage}\\\\

\noindent\begin{minipage}{\columnwidth}
\begin{multicols}{2}
\begin{lemma}
\label{lem:hadamard-involution}
\[
\InputIfFileExists{hadamard-involution.tikz}{}{\input{./figures/hadamard-involution.tikz}}
\]
\end{lemma}

\begin{lemma}
\label{lem:h-loop}
\[
\InputIfFileExists{h-loop.tikz}{}{\input{./figures/h-loop.tikz}}
\]
\end{lemma}
\end{multicols}
\end{minipage}\\\\

\noindent\begin{minipage}{\columnwidth}
\begin{multicols}{2}
\begin{lemma}
\label{lem:gn-pi_2-0-0-equals-sqrt2-exp-pi_4}
\[
\begin{tikzpicture}
	\begin{pgfonlayer}{nodelayer}
		\node [style=gn] (0) at (-1, -0) {$\frac{\pi}{2}$};
		\node [style=none] (1) at (0, -0) {=};
		\node [style=gn] (2) at (1, -0.2500001) {$\frac{\pi}{4}$};
		\node [style=rn] (3) at (1, 0.5) {$\pi$};
	\end{pgfonlayer}
	\begin{pgfonlayer}{edgelayer}
		\draw (3) to (2);
	\end{pgfonlayer}
\end{tikzpicture}}
\]
\end{lemma}

\begin{lemma}
\label{lem:euler-decomp-with-scalar}
\[
\InputIfFileExists{euler-decomp-with-scalar.tikz}{}{\input{./figures/euler-decomp-with-scalar.tikz}}
\]
\end{lemma}
\end{multicols}
\end{minipage}\\\\

\noindent\begin{minipage}{\columnwidth}
\begin{multicols}{2}
\begin{lemma}
\label{lem:C1-original}
\[
\InputIfFileExists{control-commutation-2.tikz}{}{\input{./figures/control-commutation-2.tikz}}
\]
\end{lemma}

\begin{lemma}
\label{lem:green-state-pi_2-is-red-state-minus-pi_2}
\[
\InputIfFileExists{green-state-pi_2-is-red-state-minus-pi_2.tikz}{}{\input{./figures/green-state-pi_2-is-red-state-minus-pi_2.tikz}}
\]
\end{lemma}
\end{multicols}
\end{minipage}\\\\

\noindent\begin{minipage}{\columnwidth}
\begin{multicols}{2}
\begin{lemma}
\label{lem:pi-red-state-on-triangle}
\[
\InputIfFileExists{pi-red-state-on-triangle.tikz}{}{\input{./figures/pi-red-state-on-triangle.tikz}}
\]
\end{lemma}

\begin{lemma}
\label{lem:red-state-on-upside-down-triangle}
\[
\InputIfFileExists{red-state-on-upside-down-triangle.tikz}{}{\input{./figures/red-state-on-upside-down-triangle.tikz}}
\]
\end{lemma}
\end{multicols}
\end{minipage}\\\\

\noindent\begin{minipage}{\columnwidth}
\begin{multicols}{2}
\begin{lemma}
\label{lem:supp-to-minus-pi_4}
\[
\InputIfFileExists{supp-to-minus-pi_4.tikz}{}{\input{./figures/supp-to-minus-pi_4.tikz}}
\]
\end{lemma}

\begin{lemma}
\label{lem:black-dot-swappable-outputs}
\[
\InputIfFileExists{lemma-black-dot-swappable-outputs.tikz}{}{\input{./figures/lemma-black-dot-swappable-outputs.tikz}}
\]
\end{lemma}
\end{multicols}
\end{minipage}\\\\


\begin{proof}
All these lemmas except Lemmas \ref{lem:C1-original}, \ref{lem:multiplying-global-phases} and \ref{lem:bicolor-0-alpha} come from the completeness in the \frag4 \cite{JPV}.\\
$\bullet$ \ref{lem:multiplying-global-phases}:
\def\fig{multiplying-global-phases-proof}
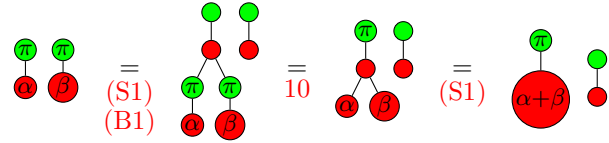
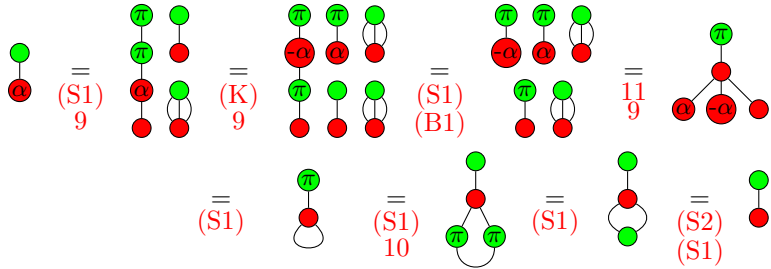
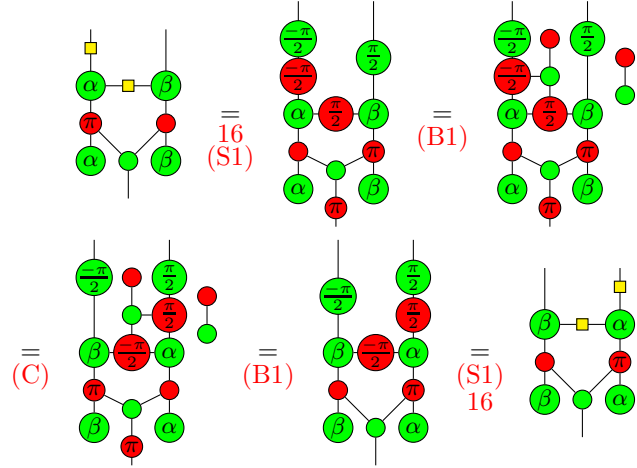
\begin{align*}

\eq{\so{}\\\bo{}}
\eq{\ref{lem:k1}}
\eq{\so{}}
\end{align*}
$\bullet$ \ref{lem:bicolor-0-alpha}:
\def\fig{bicolor-0-alpha-proof}
\begin{align*}

\eq{\so{}\\\ref{lem:inverse}}
\eq{\kt{}\\\ref{lem:inverse}}
\eq{\so{}\\\bo{}}
\eq{\ref{lem:multiplying-global-phases}\\\ref{lem:inverse}}\begin{tikzpicture}
	\begin{pgfonlayer}{nodelayer}
		\node [style=dot] (46)  at (0.0, 0.375) {};
		\node [style=dot] (47)  at (0.0, 0.125) {};
		\node [style=none] (48)  at (0.0, -0.375) {};
	\end{pgfonlayer}
	\begin{pgfonlayer}{edgelayer}
		\draw (46) to (48.center);
		\draw [style=none, in=135, out=45, loop] (46) to ();
	\end{pgfonlayer}
\end{tikzpicture}\\
\eq{\so{}}\input{./figures/\fig/\fig_05.tikz}
\eq{\so{}\\\ref{lem:k1}}\input{./figures/\fig/\fig_06.tikz}
\eq{\so{}}\input{./figures/\fig/\fig_07.tikz}
\eq{\st{}\\\so{}}\input{./figures/\fig/\fig_08.tikz}
\end{align*}
$\bullet$ \ref{lem:C1-original}:
\def\fig{control-commutation-2-proof-arxiv}
\begin{align*}

\eq{\ref{lem:euler-decomp-with-scalar}\\\so{}}
\eq{\bo{}}\\
\eq{\com{}}
\eq{\bo{}}
\eq{\so{}\\\ref{lem:euler-decomp-with-scalar}}\input{./figures/\fig/\fig_05.tikz}
\end{align*}
\end{proof}

\subsection{Proof of Lemma \ref{lem:alphas-on-X}}
\label{prf:alphas-on-X}
\def\fig{second-matrix-gn-alpha-proof-arxiv}
\begin{align*}

\eq{\h{}\\\ref{lem:euler-decomp-with-scalar}\\\so{}}
\eq{\bt{}}\\
\eq{\h{}}
\eq{\ref{lem:C1-original}}
\eq{\h{}\\\so{}}\input{./figures/\fig/\fig_05.tikz}\\
\eq{\ref{lem:hopf}}\input{./figures/\fig/\fig_06.tikz}
\eq{\ref{lem:green-state-pi_2-is-red-state-minus-pi_2}\\\so{}\\\kt{}\\\ref{lem:multiplying-global-phases}}\input{./figures/\fig/\fig_07.tikz}
\eq{\ref{lem:supp-to-minus-pi_4}\\\ref{lem:multiplying-global-phases}}\input{./figures/\fig/\fig_08.tikz}
\eq{\st{}\\\ref{lem:hopf}}\input{./figures/\fig/\fig_09.tikz}
\end{align*}
\qed

\subsection{Proof of Proposition \ref{lem:rank}}
\label{prf:rank}
We will prove the result diagrammatically. If $\vec x\in\{0;1\}^n$, we denote $u_n(\vec{x}):= \interp{\rx{$x_1\pi$}\cdots\rx{$x_n\pi$}}$. Notice that $\left(u_n(\vec x)\right)_{\vec x \in \{ 0,1 \}^n}$ forms a basis of $\mathbb{C}^{2^n}$. We show that if $01$ appears in the word $\vec x$, then $P_n\circ u_n(\vec x) = 0$. Diagrammatically, by completeness of the \frag4 of the ZX-Calculus, since the equations are sound:
\begin{align*}
\zx&\vdash\qquad 
\InputIfFileExists{kernel-of-matrix-M_n-1.tikz}{}{\input{./figures/kernel-of-matrix-M_n-1.tikz}}
\\
\zx&\vdash\qquad 
\InputIfFileExists{kernel-of-matrix-M_n-2.tikz}{}{\input{./figures/kernel-of-matrix-M_n-2.tikz}}

\end{align*}
The scalar $
\begin{tikzpicture}
	\begin{pgfonlayer}{nodelayer}
		\node [style=rn] (0) at (0, -0) {$\pi$};
	\end{pgfonlayer}
\end{tikzpicture}}
$ representing $0$, the base case is
handled by the first equality. In the general case, either $01$
appears on the first two wires, and the same equality produces the
result, otherwise the second schema appears, and $01$ appears
somewhere in the word applied to $P_{n-1}$. This proves the result by
induction.\\
Hence, the only possible words that are not in the kernel of $P_n$ are $1^p0^{n-p}$ for $p\in\{0,\cdots,n\}$, so there are $n+1$ of them.
\qed

\subsection{Details of the Proof for Corollary \ref{cor:distribution}}
\label{prf:distribution}
We first plug the basis $\left(\rx{}, \rx{$\pi$}\right)$ in the input:
\begin{itemize}
\item \rx{}
\begin{itemize}
\item Left hand side:
\def\fig{add-axiom-2-l-0-arxiv}
\begin{align*}

\eq{\bo{}\\\ref{lem:bicolor-0-alpha}}
\eq{\sth{}\\\so{}\\\ref{lem:green-state-pi_2-is-red-state-minus-pi_2}\\\st{}\\\ref{lem:2-is-sqrt2-squared}\\\ref{lem:inverse}}\\
\eq{\bt{}}
\eq{\supp{}\\\sth{}\\\so{}}
\eq{\ref{lem:inverse}\\\bo{}}\input{./figures/\fig/\fig_05.tikz}
\end{align*}
\item Right hand side:
\def\fig{add-axiom-2-r-0}
\begin{align*}

\eq{\vdots}
\eq{\bo{}\\\so{}\\\st{}\\\ref{lem:2-is-sqrt2-squared}\\\ref{lem:inverse}}
\end{align*}
\end{itemize}
The resulting two diagrams are equal when \rx{} is plugged.
\item \rx{$\pi$}
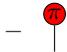
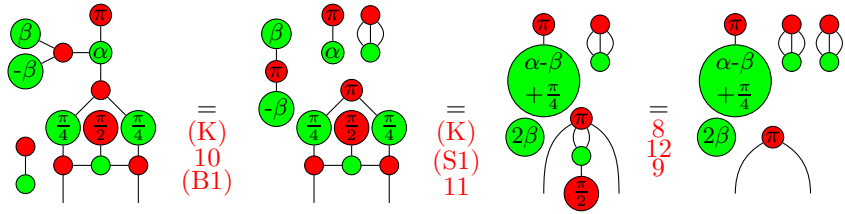
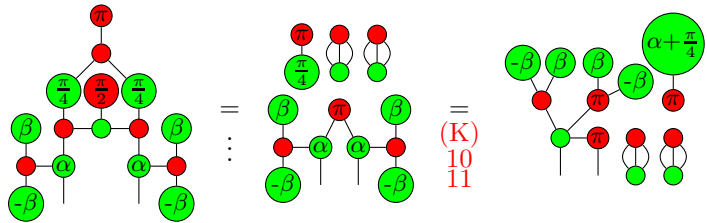
\begin{itemize}
\item Left hand side:
\def\fig{add-axiom-2-l-pi}
\begin{align*}

\eq{\kt{}\\\ref{lem:k1}\\\bo{}}
\eq{\kt{}\\\so{}\\\ref{lem:multiplying-global-phases}}
\eq{\ref{lem:hopf}\\\ref{lem:bicolor-0-alpha}\\\ref{lem:inverse}}
\end{align*}
\item Right hand side:
\def\fig{add-axiom-2-r-pi}
\begin{align*}

\eq{\vdots}
\eq{\kt{}\\\ref{lem:k1}\\\ref{lem:multiplying-global-phases}}
\end{align*}
Now we could have concluded directly with the help of Corollaries \ref{cor:gen-supp} and \ref{cor:big-scalar-equation}. For the sake of the example, though, we are going to plug our basis on, say, the left hanging branch:
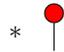
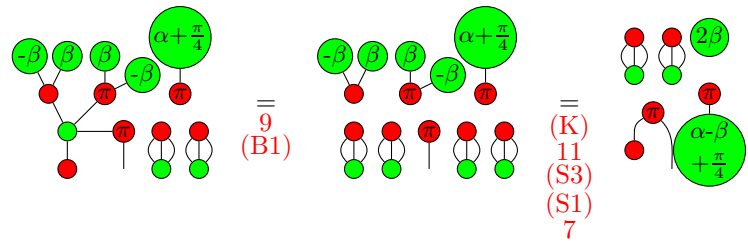
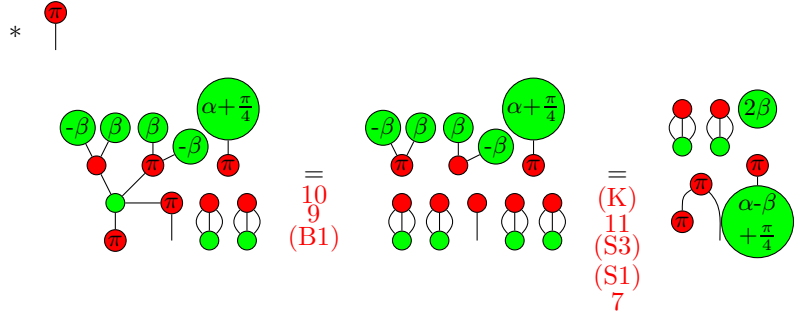
\begin{itemize}
\item \rx{}
\def\fig{add-axiom-2-r-pi-0}
\begin{align*}

\eq{\ref{lem:inverse}\\\bo{}}
\eq{\kt{}\\\ref{lem:multiplying-global-phases}\\\sth{}\\\so{}\\\ref{lem:2-is-sqrt2-squared}}
\end{align*}
\item \rx{$\pi$}
\def\fig{add-axiom-2-r-pi-pi}
\begin{align*}

\eq{\ref{lem:k1}\\\ref{lem:inverse}\\\bo{}}
\eq{\kt{}\\\ref{lem:multiplying-global-phases}\\\sth{}\\\so{}\\\ref{lem:2-is-sqrt2-squared}}
\end{align*}
\end{itemize}
\end{itemize}
\end{itemize}
Hence, the two initial diagrams result in the same diagram when the basis is applied. Thanks to Theorem \ref{thm:basis}, the \zxcalc proves the equality between the two initial diagrams.
\qed

\subsection{Details of the Proof for Corollary \ref{cor:big-scalar-equation}}
\label{prf:big-scalar-equation}~
\begin{itemize}
\item $\alpha=0$:
	\begin{itemize}
	\item Left hand side:
		\def\fig{big-scalar-equation-l-0}
		\begin{align*}
		
		\eq{\ref{lem:inverse}\\\bo{}}
		\eq{\ref{lem:2-is-sqrt2-squared}\\\ref{lem:inverse}}
		\end{align*}
	\item Right hand side:
		\def\fig{big-scalar-equation-r-0}
		\begin{align*}
		
		\eq{\kt{}\\\ref{lem:multiplying-global-phases}}
		\end{align*}
	\end{itemize}
\item $\alpha=\pi$:
	\begin{itemize}
	\item Left hand side:
		\def\fig{big-scalar-equation-l-pi-arxiv}
		\begin{align*}
		
		\eq{\kt{}\\\ref{lem:k1}\\\ref{lem:inverse}\\\bo{}}
		\eq{\ref{lem:2-is-sqrt2-squared}\\\ref{lem:inverse}}
		\end{align*}
	\item Right hand side:
		\def\fig{big-scalar-equation-r-pi}
		\begin{align*}
		
		\eq{\kt{}\\\ref{lem:multiplying-global-phases}}
		\end{align*}
	\end{itemize}
\item $\alpha=\frac{\pi}{2}$:
	\begin{itemize}
	\item Left hand side:
		\def\fig{big-scalar-equation-l-pi_2}
		\begin{align*}
		
		\eq{\ref{lem:inverse}\\\kt{}}
		\eq{\supp{}}
		\eq{\st{}\\\so{}}
		\end{align*}
	\item Right hand side:
		\def\fig{big-scalar-equation-r-pi_2-arxiv}
		\begin{align*}
		
		\eq{\kt{}\\\ref{lem:multiplying-global-phases}\\\bo{}}
		\eq{\supp{}}
		\eq{\ref{lem:inverse}\\\bo{}\\\so{}}
		\end{align*}
	\end{itemize}
\end{itemize}
The results are the same for three different values of $\alpha$. This is enough to get the equation in Corollary \ref{cor:big-scalar-equation}, according to Theorem \ref{thm:valuations}.
\qed

\subsection{Rules of the ZW-Calculus}
\label{sec:rules-zw}
~
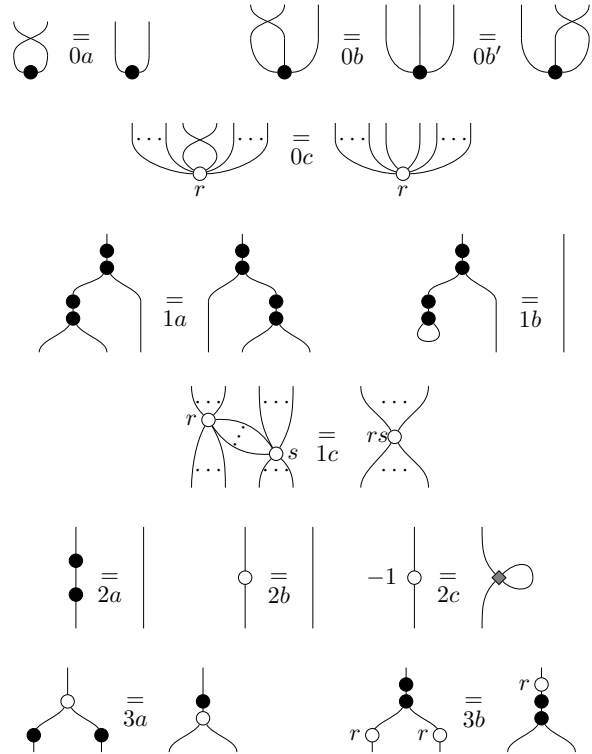
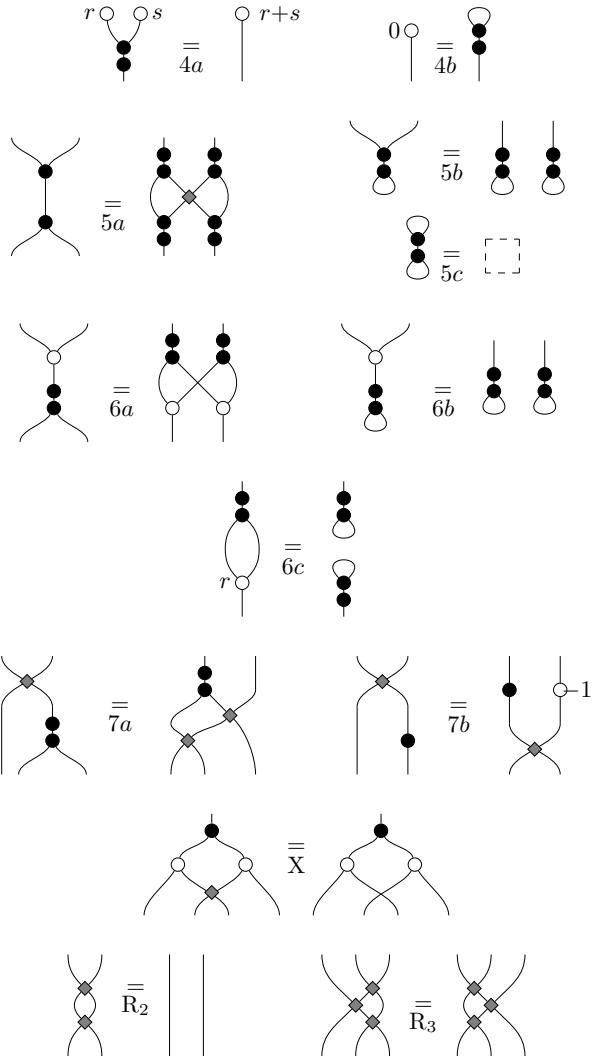
\begin{center}
\def\scale{0.9}
\scalebox{\scale}{
\InputIfFileExists{ZW-rule-0-no-braid.tikz}{}{\input{./figures/ZW-rule-0-no-braid.tikz}}
}\\~\\~\\
\scalebox{\scale}{
\InputIfFileExists{ZW-rule-1.tikz}{}{\input{./figures/ZW-rule-1.tikz}}
} \\~\\~\\
\scalebox{\scale}{
\InputIfFileExists{ZW-rule-2-no-braid.tikz}{}{\input{./figures/ZW-rule-2-no-braid.tikz}}
}\\~\\~\\
\scalebox{\scale}{
\InputIfFileExists{ZW-rule-3.tikz}{}{\input{./figures/ZW-rule-3.tikz}}
} \\~\\~\\
\scalebox{\scale}{
\InputIfFileExists{ZW-rule-4.tikz}{}{\input{./figures/ZW-rule-4.tikz}}
} \\~\\~\\
\scalebox{\scale}{
\InputIfFileExists{ZW-rule-5-no-braid.tikz}{}{\input{./figures/ZW-rule-5-no-braid.tikz}}
} \\~\\~\\
\scalebox{\scale}{
\InputIfFileExists{ZW-rule-6-a-b.tikz}{}{\input{./figures/ZW-rule-6-a-b.tikz}}
} \\~\\~\\
\scalebox{\scale}{
\InputIfFileExists{ZW-rule-6-c.tikz}{}{\input{./figures/ZW-rule-6-c.tikz}}
} \\~\\~\\
\scalebox{\scale}{
\InputIfFileExists{ZW-rule-7-no-braid.tikz}{}{\input{./figures/ZW-rule-7-no-braid.tikz}}
}\\~\\~\\
\scalebox{\scale}{
\InputIfFileExists{ZW-rule-X-no-braid.tikz}{}{\input{./figures/ZW-rule-X-no-braid.tikz}}
}\\~\\~\\
\scalebox{\scale}{
\InputIfFileExists{reidmeister-3.tikz}{}{\input{./figures/reidmeister-3.tikz}}
}
\end{center}

\subsection{Proof of Proposition \ref{prop:double-interpretation-equivalence-1}}
\label{prf:double-interp-eq}
The result is obvious for cups, caps, single wires, empty diagrams and swaps. Moreover, if we have the result for green dots and the Hadamard gate, then we also have it for red dots by construction.\\
For green dots, since $n=\max\left(0,\left\lceil \log_2(1)\right\rceil\right)=0$, $\beta = \gamma = 0$:
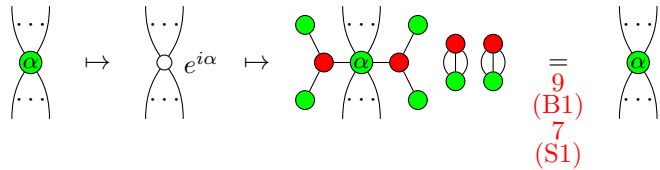
\def\fig{gn-alpha-double-interp}
\begin{align*}

~~\mapsto~~
~~\mapsto~~
\eq{\ref{lem:inverse}\\\bo{-c}\\\ref{lem:2-is-sqrt2-squared}\\\so{-c}}
\end{align*}
For Hadamard, first notice:
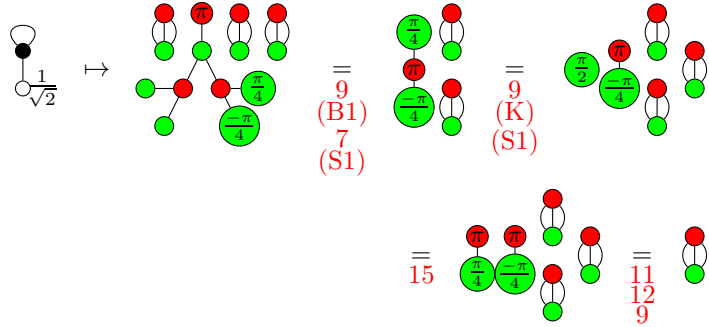
\def\fig{1-over-sqrt2-ZW-to-ZX}
\begin{align*}

~~\mapsto~~
\eq{\ref{lem:inverse}\\\bo{-c}\\\ref{lem:2-is-sqrt2-squared}\\\so{-c}}
\eq{\ref{lem:inverse}\\\kt{-c}\\\so{-c}}\\
\eq{\ref{lem:gn-pi_2-0-0-equals-sqrt2-exp-pi_4}}
\eq{\ref{lem:multiplying-global-phases}\\\ref{lem:bicolor-0-alpha}\\\ref{lem:inverse}}\input{./figures/\fig/\fig_05.tikz}
\end{align*}
since $n=0$, $\beta=\arccos{1/\sqrt{2}}=\pi/4$, $\gamma=\arccos{1}=0$. Finally:
\def\fig{hadamard-double-interpretation}
\begin{align*}
\quad\mapsto\quad\quad\mapsto\quad
\eq[\quad]{\so{}\\\st{}\\\ref{lem:inverse}}
\end{align*}
\qed

\subsection{Lemmas for $\zx_c$}

\begin{lemma}
\label{lem:prod-cos}
\[\zx_c\vdash
\InputIfFileExists{add-axiom.tikz}{}{\input{./figures/add-axiom.tikz}}
\]
\end{lemma}
\begin{proof}
\def\fig{add-axiom-3-to-A1}
\begin{align*}
\zx_c\vdash~~
\eq{Thm~\ref{thm:provability}}\\
\eq{\add{-c}}
\eq{\ref{lem:supp-to-minus-pi_4}\\\ref{lem:multiplying-global-phases}\\\ref{lem:bicolor-0-alpha}\\\ref{lem:inverse}}
\end{align*}
where $\cos(\gamma) = \frac{1}{2}(\cos(\alpha-\beta)+ \cos(\alpha+\beta)) = \cos(\alpha)\cos(\beta)$.
\end{proof}

\begin{lemma}
\label{lem:add-bis}
We can deduce an equality similar to the rule \add{-c}:
\[\zx_c\vdash
\InputIfFileExists{add-axiom-3-simplified-2.tikz}{}{\input{./figures/add-axiom-3-simplified-2.tikz}}
\]
\end{lemma}
\begin{proof}
\def\fig{add-axiom-3-simp-2-proof}
\begin{align*}
\zx_c\vdash~~
\eq{\ref{lem:supp-to-minus-pi_4}\\\kt{-c}\\\ref{lem:multiplying-global-phases}}
\eq{\ref{lem:prod-cos}}\\
\eq{\ref{lem:prod-cos}}
\eq{\kt{-c}\\\ref{lem:multiplying-global-phases}}
\eq{\add{-c}}\input{./figures/\fig/\fig_05.tikz}
\end{align*}
where $\cos(\gamma')=\frac{\cos(\gamma)}{\cos(\frac{\pi}{4})}=\sqrt{2}\cos(\gamma)$ and $\cos(\gamma'')=\sqrt{2}\cos(\gamma')=2\cos(\gamma)$. We end up with the right part of the rule \add{-c}, and applying the rule with $\cos(\gamma'')$ gives the wanted condition on the angles.
\end{proof}

\begin{lemma}
\label{lem:white-dot-to-zx-generalised-form}~\\
Let $\rho\in\mathbb{R}+$. Then, for any $n_1,n_2\geq\max\left(0,\left\lceil \log_2(\rho)\right\rceil\right)$:
\[\zx_c\vdash~~\scalebox{0.91}{
\InputIfFileExists{ZW-white-dot-to-ZX-generalised.tikz}{}{\input{./figures/ZW-white-dot-to-ZX-generalised.tikz}}
}\]
\end{lemma}

\begin{proof}First we prove:
\def\fig{lemma-pi_3-branches-arxiv}
\begin{align*}
\zx_c\vdash~~
\eq{\ref{cor:gen-supp}\\\ref{lem:inverse}\\\ref{lem:hopf}}
\eq{\so{-c}\\\ref{lem:prod-cos}}\\
\eq{\so{-c}\\\ref{lem:inverse}\\\ref{lem:k1}\\\ref{lem:multiplying-global-phases}}
\eq{\ref{lem:supp-to-minus-pi_4}}
\eq{\ref{lem:2-is-sqrt2-squared}\\\ref{lem:multiplying-global-phases}\\\ref{lem:inverse}}\input{./figures/\fig/\fig_05.tikz}
\end{align*}
We now show the result for $n\geq\max\left(0,\left\lceil \log_2(\rho)\right\rceil\right)$ and $n+1$, which then generalises to lemma \ref{lem:white-dot-to-zx-generalised-form} by induction:
\def\fig{ZW-white-dot-to-ZX-generalised-proof}
\begin{align*}
\zx_c\vdash~~
\eq{}\\
\eq{\ref{lem:prod-cos}}
\end{align*}
with:
\begin{align*}
\beta &= \arccos{\frac{\rho}{2^n}}\qquad
\gamma = \arccos{\frac{1}{2^n}}\\
\beta' &= \arccos{\frac{\rho}{2^n}\cos(\pi/3)}=\arccos{\frac{\rho}{2^{n+1}}}\\
\gamma' &= \arccos{\frac{1}{2^{n+1}}}\\
\end{align*}
\end{proof}

\begin{corollary}
\label{cor:arccos-2^-n}
For any $n\in\mathbb{N}$, with $\gamma=\arccos{\frac{1}{2^n}}$:
\def\fig{lemma-arccos-1_2-to-the-n-branches}
\begin{align*}
\zx_c\vdash~~
\eq{}
\eq{}
\end{align*}
\end{corollary}

\begin{lemma}
\label{lem:gn-inverse-c}
The green node \begin{tikzpicture}
	\begin{pgfonlayer}{nodelayer}
		\node [style=gn] (0) at (0,0) {$\alpha$};
	\end{pgfonlayer}
\end{tikzpicture} 
has an inverse if $\alpha\neq\pi\mod 2\pi$:
\[
\InputIfFileExists{gn-alpha-inverse-no-triangle.tikz}{}{\input{./figures/gn-alpha-inverse-no-triangle.tikz}}
\]
for $n\geq \log_2\left(\frac{1}{|cos(\alpha/2)|}\right)$ and $\beta = 2\arccos{\frac{1}{2^n\cos(\alpha/2)}}$.
\end{lemma}

\begin{proof}
Notice that $\beta$ is well defined if $\alpha\neq\pi\mod 2\pi$. With these values of $n$ and $\beta$, $\cos(\alpha/2)\cos(\beta/2)=\cos(\gamma)$ with $\gamma = \arccos{\frac{1}{2^n}}$. Then:
\def\fig{gn-alpha-inverse-no-triangle-proof}
\begin{align*}

\eq{\so{-c}\\\kt{-c}\\\ref{lem:multiplying-global-phases}\\\ref{lem:inverse}}
\eq{\bo{-c}\\\ref{lem:k1}\\\ref{lem:inverse}}\\
\eq{\ref{lem:prod-cos}}
\eq{\ref{lem:inverse}\\\ref{lem:2-is-sqrt2-squared}\\\so{-c}\\\sth{-c}}
\eq{\bo{-c}\\\ref{lem:inverse}}\input{./figures/\fig/\fig_05.tikz}\\
\eq{\ref{cor:arccos-2^-n}}\input{./figures/\fig/\fig_06.tikz}
\eq{\ref{lem:inverse}}\input{./figures/\fig/\fig_07.tikz}
\end{align*}
\end{proof}

\subsection{Proof of Proposition \ref{prop:X-is-homomorphism}}
\label{prf:X-is-homomorphism}
Since we have built the set of rules $\zx_c$ upon the one in \cite{JPV} which is complete for Clifford+T, we basically just need to prove the result for the ZW-rules in which a parameter (different from $\pm 1$) appears: $1c$, $3b$, $4a$, $4b$ and $6c$. Notice that the rule $0c$ is obvious.\\
$\bullet$ $1c$:
\def\fig{rule-1c-proof}
\begin{align*}

~~\mapsto~~
\eq{\ref{lem:prod-cos}}
\end{align*}
where:
\begin{align*}
n_k=\max\left(0,\left\lceil \log_2(\rho_k)\right\rceil\right)\qquad&\qquad
n=n_1+n_2\\
\beta_k = \arccos{\frac{\rho_k}{2^{n_k}}}\qquad&\qquad
\beta = \arccos{\frac{\rho}{2^n}}\\
\gamma_k = \arccos{\frac{1}{2^{n_k}}}\qquad&\qquad
\gamma = \arccos{\frac{1}{2^n}}
\end{align*}
Notice that $\left\lceil \log_2(\rho_1\rho_2)\right\rceil=\left\lceil \log_2(\rho_1)+\log_2(\rho_2)\right\rceil \leq \left\lceil \log_2(\rho_1)\right\rceil+\left\lceil \log_2(\rho_2)\right\rceil$, so the result might not be precisely the one given by $\left[
\InputIfFileExists{white-dot-rho1-rho2.tikz}{}{\input{./figures/white-dot-rho1-rho2.tikz}}
\right]_X$, but it can be patched thanks to lemma \ref{lem:white-dot-to-zx-generalised-form}.\\
$\bullet$ $3b$: corollary \ref{cor:distribution}.\\
$\bullet$ $4a$: suppose $\rho_1\geq\rho_2$, then using lemma \ref{lem:white-dot-to-zx-generalised-form} to have the same $n$ on both sides:
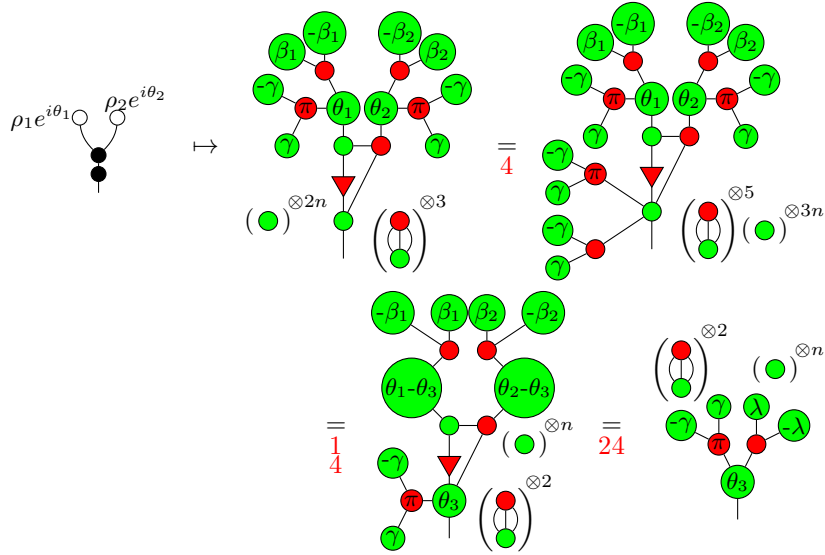
\def\fig{rule-4a-proof}
\begin{align*}

~~\mapsto~~
\eq{\ref{cor:arccos-2^-n}}\\
\eq{\ref{cor:distribution}\\\ref{cor:arccos-2^-n}}
\eq{\ref{lem:add-bis}}
\end{align*}
with 
\begin{align*}
\beta_k &= \arccos{\frac{\rho_k}{2^n}}\qquad
\gamma = \arccos{\frac{1}{2^n}}\\
\theta_3 &=\arg(\rho_1e^{i\theta_1}+\rho_2e^{i\theta_2})\\
\lambda &= \arccos{e^{i\theta_1-\theta_3}\cos{\beta_1}+e^{i\theta_2-\theta_3}\cos{\beta_2}}\\
 &= \arccos{\frac{\rho_1e^{i\theta_1}+\rho_2e^{i\theta_2}}{e^{i\theta_3}2^n}}
\end{align*}
which is what $\left[
\begin{tikzpicture}
	\begin{pgfonlayer}{nodelayer}
		\node [anchor=west, style=none] (0) at (0.25, 0.5) {\footnotesize $\rho_1e^{i\theta_1}+\rho_2e^{i\theta_2}$};
		\node [style=white dot] (1) at (0, 0.5) {};
		\node [style=none] (2) at (0, -0.5) {};
	\end{pgfonlayer}
	\begin{pgfonlayer}{edgelayer}
		\draw (1) to (2.center);
	\end{pgfonlayer}
\end{tikzpicture}}
\right]_X$ gives.\\
$\bullet$ $4b$:
\def\fig{rule-4b-proof}
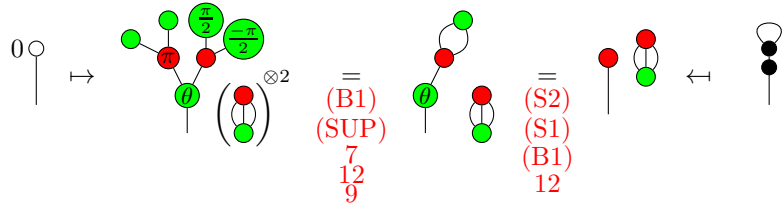
\begin{align*}

~~\mapsto~~
\eq{\bo{-c}\\\supp{-c}\\\ref{lem:2-is-sqrt2-squared}\\\ref{lem:bicolor-0-alpha}\\\ref{lem:inverse}}
\eq{\st{-c}\\\so{-c}\\\bo{-c}\\\ref{lem:bicolor-0-alpha}}
~~\mapsfrom~~
\end{align*}
$\bullet$ $6c$: First, using corollary \ref{cor:arccos-2^-n}, for $\gamma = \arccos{\frac{1}{2^n}}$,
\def\fig{lemma-for-rule-6c}
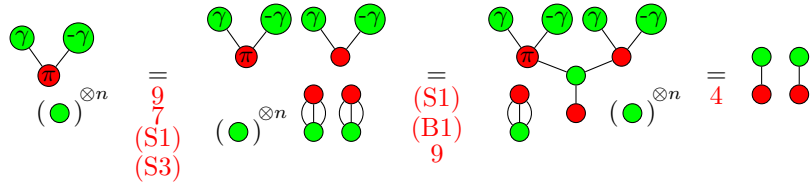
\begin{align*}

\eq{\ref{lem:inverse}\\\ref{lem:2-is-sqrt2-squared}\\\so{-c}\\\sth{-c}}
\eq{\so{-c}\\\bo{-c}\\\ref{lem:inverse}}
\eq{\ref{cor:arccos-2^-n}}
\end{align*}
then, with:
\begin{align*}
n &=\max\left(0,\left\lceil \log_2(\rho)\right\rceil\right)\\
\beta &= \arccos{\frac{\rho}{2^{n}}}\\
\gamma &= \arccos{\frac{1}{2^{n}}}
\end{align*}
\def\fig{rule-6c-proof-2}
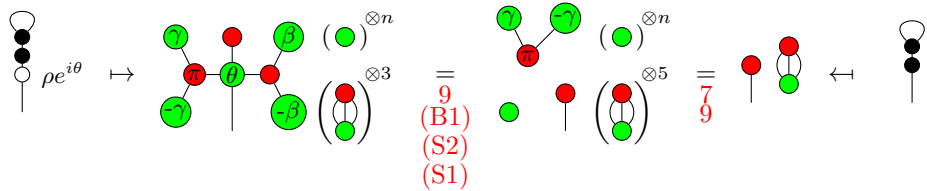
\begin{align*}

~~\mapsto~~
\eq{\ref{lem:inverse}\\\bo{-c}\\\st{-c}\\\so{-c}}
\eq{\ref{lem:2-is-sqrt2-squared}\\\ref{lem:inverse}}
~~\mapsfrom~~
\end{align*}
This is enough to show that rule $6c$ stands, because the case $r=1$ has already been treated to show the completeness of $\zx$ for Clifford+T.
\qed


\begin{thebibliography}{10}
\providecommand{\bibitemdeclare}[2]{}
\providecommand{\surnamestart}{}
\providecommand{\surnameend}{}
\providecommand{\urlprefix}{Available at }
\providecommand{\url}[1]{\texttt{#1}}
\providecommand{\href}[2]{\texttt{#2}}
\providecommand{\urlalt}[2]{\href{#1}{#2}}
\providecommand{\doi}[1]{doi:\urlalt{http://dx.doi.org/#1}{#1}}
\providecommand{\bibinfo}[2]{#2}

\bibitemdeclare{article}{pi_2-complete}
\bibitem{pi_2-complete}
\bibinfo{author}{Miriam \surnamestart Backens\surnameend}
  (\bibinfo{year}{2014}): \emph{\bibinfo{title}{The {ZX}-calculus is complete
  for stabilizer quantum mechanics}}.
\newblock {\sl \bibinfo{journal}{New Journal of Physics}}
  \bibinfo{volume}{16}(\bibinfo{number}{9}), p. \bibinfo{pages}{093021}.
\newblock \urlprefix\url{http://stacks.iop.org/1367-2630/16/i=9/a=093021}.

\bibitemdeclare{article}{toy-model-graph}
\bibitem{toy-model-graph}
\bibinfo{author}{Miriam \surnamestart Backens\surnameend} \&
  \bibinfo{author}{Ali~Nabi \surnamestart Duman\surnameend}
  (\bibinfo{year}{2014}): \emph{\bibinfo{title}{A complete graphical calculus
  for Spekkens' toy bit theory}}.
\newblock {\sl \bibinfo{journal}{Foundations of Physics}}, pp.
  \bibinfo{pages}{1--34}, \doi{10.1007/s10701-015-9957-7}.

\bibitemdeclare{inproceedings}{de2017zx}
\bibitem{de2017zx}
\bibinfo{author}{Niel \surnamestart de~Beaudrap\surnameend} \&
  \bibinfo{author}{Dominic \surnamestart Horsman\surnameend}
  (\bibinfo{year}{2017}): \emph{\bibinfo{title}{The ZX calculus is a language
  for surface code lattice surgery}}.
\newblock In: {\sl \bibinfo{booktitle}{QPL 2017}}.
\newblock \urlprefix\url{https://arxiv.org/abs/1704.08670}.

\bibitemdeclare{unpublished}{chancellor2016coherent}
\bibitem{chancellor2016coherent}
\bibinfo{author}{Nicholas \surnamestart Chancellor\surnameend},
  \bibinfo{author}{Aleks \surnamestart Kissinger\surnameend},
  \bibinfo{author}{Joschka \surnamestart Roffe\surnameend},
  \bibinfo{author}{Stefan \surnamestart Zohren\surnameend} \&
  \bibinfo{author}{Dominic \surnamestart Horsman\surnameend}
  (\bibinfo{year}{2016}): \emph{\bibinfo{title}{Graphical Structures for Design
  and Verification of Quantum Error Correction}}.
\newblock \urlprefix\url{https://arxiv.org/abs/1611.08012}.

\bibitemdeclare{article}{interacting}
\bibitem{interacting}
\bibinfo{author}{Bob \surnamestart Coecke\surnameend} \& \bibinfo{author}{Ross
  \surnamestart Duncan\surnameend} (\bibinfo{year}{2011}):
  \emph{\bibinfo{title}{Interacting quantum observables: categorical algebra
  and diagrammatics}}.
\newblock {\sl \bibinfo{journal}{New Journal of Physics}}
  \bibinfo{volume}{13}(\bibinfo{number}{4}), p. \bibinfo{pages}{043016}.
\newblock \urlprefix\url{http://stacks.iop.org/1367-2630/13/i=4/a=043016}.

\bibitemdeclare{inbook}{ghz-w}
\bibitem{ghz-w}
\bibinfo{author}{Bob \surnamestart Coecke\surnameend} \& \bibinfo{author}{Aleks
  \surnamestart Kissinger\surnameend} (\bibinfo{year}{2010}):
  \emph{\bibinfo{title}{The Compositional Structure of Multipartite Quantum
  Entanglement}}, pp. \bibinfo{pages}{297--308}.
\newblock \bibinfo{publisher}{Springer Berlin Heidelberg},
  \bibinfo{address}{Berlin, Heidelberg}, \doi{10.1007/978-3-642-14162-1\_25}.
\newblock \urlprefix\url{http://dx.doi.org/10.1007/978-3-642-14162-1\_25}.

\bibitemdeclare{book}{picturing-qp}
\bibitem{picturing-qp}
\bibinfo{author}{Bob \surnamestart Coecke\surnameend} \& \bibinfo{author}{Aleks
  \surnamestart Kissinger\surnameend} (\bibinfo{year}{2017}):
  \emph{\bibinfo{title}{Picturing Quantum Processes: A First Course in Quantum
  Theory and Diagrammatic Reasoning}}.
\newblock \bibinfo{publisher}{Cambridge University Press},
  \doi{10.1017/9781316219317}.

\bibitemdeclare{inbook}{duncan2013mbqc}
\bibitem{duncan2013mbqc}
\bibinfo{author}{Ross \surnamestart Duncan\surnameend} (\bibinfo{year}{2013}):
  \emph{\bibinfo{title}{A graphical approach to measurement-based quantum
  computing}}.
\newblock \doi{10.1093/acprof:oso/9780199646296.003.0003}.

\bibitemdeclare{inproceedings}{duncan2016hopf}
\bibitem{duncan2016hopf}
\bibinfo{author}{Ross \surnamestart Duncan\surnameend} \&
  \bibinfo{author}{Kevin \surnamestart Dunne\surnameend}
  (\bibinfo{year}{2016}): \emph{\bibinfo{title}{Interacting Frobenius Algebras
  Are Hopf}}.
\newblock In: {\sl \bibinfo{booktitle}{Proceedings of the 31st Annual ACM/IEEE
  Symposium on Logic in Computer Science}}, \bibinfo{series}{LICS 2016},
  \bibinfo{publisher}{ACM}, \bibinfo{address}{New York, NY, USA}, pp.
  \bibinfo{pages}{535--544}, \doi{10.1145/2933575.2934550}.
\newblock \urlprefix\url{http://doi.acm.org/10.1145/2933575.2934550}.

\bibitemdeclare{unpublished}{verifying-color-code}
\bibitem{verifying-color-code}
\bibinfo{author}{Ross \surnamestart Duncan\surnameend} \& \bibinfo{author}{Liam
  \surnamestart Garvie\surnameend} (\bibinfo{year}{2017}):
  \emph{\bibinfo{title}{Verifying the Smallest Interesting Colour Code with
  Quantomatic}}.
\newblock \urlprefix\url{https://arxiv.org/abs/1706.02717}.

\bibitemdeclare{article}{duncan2014verifying}
\bibitem{duncan2014verifying}
\bibinfo{author}{Ross \surnamestart Duncan\surnameend} \&
  \bibinfo{author}{Maxime \surnamestart Lucas\surnameend}
  (\bibinfo{year}{2014}): \emph{\bibinfo{title}{Verifying the Steane code with
  Quantomatic}}.
\newblock {\sl \bibinfo{journal}{Electronic Proceedings in Theoretical Computer
  Science}} \bibinfo{volume}{171}, pp. \bibinfo{pages}{33--49},
  \doi{10.4204/EPTCS.171.4}.

\bibitemdeclare{article}{mbqc}
\bibitem{mbqc}
\bibinfo{author}{Ross \surnamestart Duncan\surnameend} \&
  \bibinfo{author}{Simon \surnamestart Perdrix\surnameend}
  (\bibinfo{year}{2010}): \emph{\bibinfo{title}{Rewriting measurement-based
  quantum computations with generalised flow}}.
\newblock {\sl \bibinfo{journal}{Lecture Notes in Computer Science}}
  \bibinfo{volume}{6199}, pp. \bibinfo{pages}{285--296},
  \doi{10.1007/978-3-642-14162-1\_24}.
\newblock
  \urlprefix\url{http://personal.strath.ac.uk/ross.duncan/papers/gflow.pdf}.

\bibitemdeclare{inproceedings}{pivoting}
\bibitem{pivoting}
\bibinfo{author}{Ross \surnamestart Duncan\surnameend} \&
  \bibinfo{author}{Simon \surnamestart Perdrix\surnameend}
  (\bibinfo{year}{2013}): \emph{\bibinfo{title}{Pivoting makes the
  {ZX}-calculus complete for real stabilizers}}.
\newblock In: {\sl \bibinfo{booktitle}{QPL 2013}}, \doi{10.4204/EPTCS.171.5}.

\bibitemdeclare{inproceedings}{zw}
\bibitem{zw}
\bibinfo{author}{Amar \surnamestart Hadzihasanovic\surnameend}
  (\bibinfo{year}{2015}): \emph{\bibinfo{title}{A Diagrammatic Axiomatisation
  for Qubit Entanglement}}.
\newblock In: {\sl \bibinfo{booktitle}{2015 30th Annual ACM/IEEE Symposium on
  Logic in Computer Science}}, pp. \bibinfo{pages}{573--584},
  \doi{10.1109/LICS.2015.59}.

\bibitemdeclare{phdthesis}{Amar}
\bibitem{Amar}
\bibinfo{author}{Amar \surnamestart Hadzihasanovic\surnameend}
  (\bibinfo{year}{2017}): \emph{\bibinfo{title}{The algebra of entanglement and
  the geometry of composition}}.
\newblock Ph.D. thesis, \bibinfo{school}{University of Oxford}.
\newblock \urlprefix\url{https://arxiv.org/abs/1709.08086}.

\bibitemdeclare{article}{horsman2011quantum}
\bibitem{horsman2011quantum}
\bibinfo{author}{Clare \surnamestart Horsman\surnameend}
  (\bibinfo{year}{2011}): \emph{\bibinfo{title}{Quantum picturalism for
  topological cluster-state computing}}.
\newblock {\sl \bibinfo{journal}{New Journal of Physics}}
  \bibinfo{volume}{13}(\bibinfo{number}{9}), p. \bibinfo{pages}{095011}.
\newblock \urlprefix\url{http://stacks.iop.org/1367-2630/13/i=9/a=095011}.

\bibitemdeclare{unpublished}{JPV}
\bibitem{JPV}
\bibinfo{author}{Emmanuel \surnamestart Jeandel\surnameend},
  \bibinfo{author}{Simon \surnamestart Perdrix\surnameend} \&
  \bibinfo{author}{Renaud \surnamestart Vilmart\surnameend}
  (\bibinfo{year}{2017}): \emph{\bibinfo{title}{{A Complete Axiomatisation of
  the ZX-Calculus for Clifford+T Quantum Mechanics}}}.
\newblock \urlprefix\url{https://hal.archives-ouvertes.fr/hal-01529623}.
\newblock \bibinfo{note}{Working paper or preprint}.

\bibitemdeclare{inproceedings}{cyclo}
\bibitem{cyclo}
\bibinfo{author}{Emmanuel \surnamestart Jeandel\surnameend},
  \bibinfo{author}{Simon \surnamestart Perdrix\surnameend},
  \bibinfo{author}{Renaud \surnamestart Vilmart\surnameend} \&
  \bibinfo{author}{Quanlong \surnamestart Wang\surnameend}
  (\bibinfo{year}{2017}): \emph{\bibinfo{title}{{ZX-Calculus: Cyclotomic
  Supplementarity and Incompleteness for Clifford+T Quantum Mechanics}}}.
\newblock In \bibinfo{editor}{Kim~G. \surnamestart Larsen\surnameend},
  \bibinfo{editor}{Hans~L. \surnamestart Bodlaender\surnameend} \&
  \bibinfo{editor}{Jean-Francois \surnamestart Raskin\surnameend}, editors:
  {\sl \bibinfo{booktitle}{42nd International Symposium on Mathematical
  Foundations of Computer Science (MFCS 2017)}}, {\sl \bibinfo{series}{Leibniz
  International Proceedings in Informatics (LIPIcs)}}~\bibinfo{volume}{83},
  \bibinfo{publisher}{Schloss Dagstuhl--Leibniz-Zentrum fuer Informatik},
  \bibinfo{address}{Dagstuhl, Germany}, pp. \bibinfo{pages}{11:1--11:13},
  \doi{10.4230/LIPIcs.MFCS.2017.11}.
\newblock \urlprefix\url{http://drops.dagstuhl.de/opus/volltexte/2017/8117}.

\bibitemdeclare{misc}{quanto}
\bibitem{quanto}
\bibinfo{author}{A.~\surnamestart Kissinger\surnameend},
  \bibinfo{author}{L.~\surnamestart Dixon\surnameend},
  \bibinfo{author}{R.~\surnamestart Duncan\surnameend},
  \bibinfo{author}{B.~\surnamestart Frot\surnameend},
  \bibinfo{author}{A.~\surnamestart Merry\surnameend},
  \bibinfo{author}{D.~\surnamestart Quick\surnameend},
  \bibinfo{author}{M.~\surnamestart Soloviev\surnameend} \&
  \bibinfo{author}{V.~\surnamestart Zamdzhiev\surnameend}
  (\bibinfo{year}{2011}): \emph{\bibinfo{title}{Quantomatic}}.
\newblock \urlprefix\url{https://sites.google.com/site/quantomatic/}.

\bibitemdeclare{inproceedings}{kissinger2015quantomatic}
\bibitem{kissinger2015quantomatic}
\bibinfo{author}{Aleks \surnamestart Kissinger\surnameend} \&
  \bibinfo{author}{Vladimir \surnamestart Zamdzhiev\surnameend}
  (\bibinfo{year}{2015}): \emph{\bibinfo{title}{Quantomatic: A Proof Assistant
  for Diagrammatic Reasoning}}.
\newblock In \bibinfo{editor}{Amy~P. \surnamestart Felty\surnameend} \&
  \bibinfo{editor}{Aart \surnamestart Middeldorp\surnameend}, editors: {\sl
  \bibinfo{booktitle}{Automated Deduction - CADE-25}},
  \bibinfo{publisher}{Springer International Publishing},
  \bibinfo{address}{Cham}, pp. \bibinfo{pages}{326--336},
  \doi{10.1007/978-3-319-21401-6\_22}.

\bibitemdeclare{unpublished}{NgWang}
\bibitem{NgWang}
\bibinfo{author}{Kang~Feng \surnamestart Ng\surnameend} \&
  \bibinfo{author}{Quanlong \surnamestart Wang\surnameend}
  (\bibinfo{year}{2017}): \emph{\bibinfo{title}{A universal completion of the
  ZX-calculus}}.
\newblock \urlprefix\url{https://arxiv.org/abs/1706.09877}.

\bibitemdeclare{inproceedings}{incompleteness}
\bibitem{incompleteness}
\bibinfo{author}{Christian \surnamestart Schr\"oder~de Witt\surnameend} \&
  \bibinfo{author}{Vladimir \surnamestart Zamdzhiev\surnameend}
  (\bibinfo{year}{2014}): \emph{\bibinfo{title}{The {ZX}-calculus is incomplete
  for quantum mechanics}}.
\newblock In: {\sl \bibinfo{booktitle}{QPL 2014}}, \doi{10.4204/EPTCS.172.20}.

\end{thebibliography}
\end{document}